\RequirePackage{fix-cm}
\documentclass[english]{article}
\usepackage[T1]{fontenc}
\usepackage[latin9]{inputenc}
\usepackage{geometry}
\usepackage[active]{srcltx}
\usepackage{color}
\usepackage{babel}
\usepackage{array}
\usepackage{float}
\usepackage{amsthm}
\usepackage{amsmath}
\usepackage{amssymb}
\usepackage{fixltx2e}
\usepackage{graphicx}
\usepackage{esint}
\usepackage[authoryear]{natbib}
\usepackage[unicode=true,pdfusetitle,
 bookmarks=true,bookmarksnumbered=true,bookmarksopen=true,bookmarksopenlevel=1,
 breaklinks=true,pdfborder={0 0 1},backref=false,colorlinks=true]
 {hyperref}
\hypersetup{
 linkcolor=blue, citecolor=blue, urlcolor=blue, filecolor=blue,pdfpagelayout=OneColumn, pdfnewwindow=true,pdfstartview=XYZ, plainpages=false, pdfpagelabels, hyperindex=true}
\usepackage{breakurl}

\makeatletter

\providecommand{\tabularnewline}{\\}
\newcommand{\lyxdot}{.}

\numberwithin{equation}{section}
\numberwithin{figure}{section}
\usepackage{enumitem}		
\numberwithin{table}{section}
  \theoremstyle{plain}
  \newtheorem{thm}{\protect\theoremname}[section]
  \theoremstyle{remark}
  \newtheorem{rem}{\protect\remarkname}[section]
  \theoremstyle{plain}
  \newtheorem{lem}{\protect\lemmaname}[section]

\usepackage{fancyhdr}
\usepackage{lastpage}
\usepackage{amsmath}
\usepackage{graphicx}
\usepackage{color}
\usepackage{fancybox} 
\usepackage{moreverb} 
\usepackage{listings} 
\usepackage{backref}
\usepackage{fixltx2e}

\usepackage{ifpdf} 
\ifpdf 

 \IfFileExists{lmodern.sty}{\usepackage{lmodern}}{}

\fi 

\let\myTOC\tableofcontents
\renewcommand\tableofcontents{%
  \pdfbookmark[1]{\contentsname}{}
  \myTOC
}

\def\LyX{\texorpdfstring{%
  L\kern-.1667em\lower.25em\hbox{Y}\kern-.125emX\@}
  {LyX}}

\@addtoreset{footnote}{section}

\renewcommand*{\backref}[1]{}
\renewcommand*{\backrefalt}[4]{%
   \ifcase #1 
    \or 
      (Cited on page~#2)%
   \else
      (Cited on pages~#2)
    \fi} 

\usepackage{dsfont}

\usepackage{eurosym}

\definecolor{gr1}{rgb}{0.33333333,  0.65882353,  0.40784314}
\definecolor{gr2}{rgb}{0.52            ,  0.6854902  ,  0.42666667}
\definecolor{gr3}{rgb}{0.70666667,  0.71215686,  0.4454902  }
\definecolor{gr4}{rgb}{0.79372549,  0.64156863,  0.42823529}
\definecolor{gr5}{rgb}{0.78117647,  0.47372549,  0.37490196}
\definecolor{gr6}{rgb}{0.76862745,  0.30588235,  0.32156863}

\@ifundefined{showcaptionsetup}{}{%
 \PassOptionsToPackage{caption=false}{subfig}}
\usepackage{subfig}
\makeatother

  \providecommand{\lemmaname}{Lemma}
  \providecommand{\remarkname}{Remark}
\providecommand{\theoremname}{Theorem}

\begin{document}

\title{\noindent Switching to non-affine stochastic volatility:\\
A closed-form expansion for the Inverse Gamma model}

\author{{\normalsize{}\vspace{-2mm}
}\\
Nicolas Langren\'e\\
\texttt{\textsc{\small{}The Commonwealth Scientific and Industrial
Research Organisation}}\texttt{\small{}}~\\
\textsl{\small{}Real Options and Financial Risk}{\normalsize{}}\\
\texttt{\normalsize{}\href{mailto:nicolas.langrene@csiro.au}{nicolas.langrene@csiro.au}}{\normalsize{}\vspace{3mm}
}\\
Geoffrey Lee\\
\texttt{\textsc{\small{}The Commonwealth Scientific and Industrial
Research Organisation}}\texttt{\small{}}~\\
\textsl{\small{}Real Options and Financial Risk}\texttt{\small{}}~\\
\texttt{\normalsize{}\href{mailto:geoffrey.lee@csiro.au}{geoffrey.lee@csiro.au}}{\normalsize{}\vspace{3mm}
}\\
Zili Zhu\\
\texttt{\textsc{\small{}The Commonwealth Scientific and Industrial
Research Organisation}}\texttt{\small{}}~\\
\textsl{\small{}Real Options and Financial Risk}\texttt{\small{}}~\\
\texttt{\normalsize{}\href{mailto:zili.zhu@csiro.au}{zili.zhu@csiro.au}}{\normalsize{}\vspace{1mm}
}}

\date{{\normalsize{}\vspace{1mm}
First version: July 8, 2015}\\
{\normalsize{}This revised version: March 18, 2016}}
\maketitle
\begin{abstract}
\noindent This paper introduces the Inverse Gamma (IGa) stochastic
volatility model with time-dependent parameters, defined by the volatility
dynamics $dV_{t}=\kappa_{t}\left(\theta_{t}-V_{t}\right)dt+\lambda_{t}V_{t}dB_{t}$.\vspace{1mm}
\\
This non-affine model is much more realistic than classical affine
models like the Heston stochastic volatility model, even though both
are as parsimonious (only four stochastic parameters). Indeed, it
provides more realistic volatility distribution and volatility paths,
which translate in practice into more robust calibration and better
hedging accuracy, explaining its popularity among practitioners.\vspace{1mm}
\\
In order to price vanilla options with IGa volatility, we propose
a closed-form volatility-of-volatility expansion. Specifically, the
price of a European put option with IGa volatility is approximated
by a Black-Scholes price plus a weighted combination of Black-Scholes
greeks, where the weights depend only on the four time-dependent parameters
of the model.\vspace{1mm}
\\
This closed-form pricing method allows for very fast pricing and calibration
to market data. The overall quality of the approximation is very good,
as shown by several calibration tests on real-world market data where
expansion prices are compared favorably with Monte Carlo simulation
results.\vspace{1mm}
\\
This paper shows that the IGa model is as simple, more realistic,
easier to implement and faster to calibrate than classical transform-based
affine models. We therefore hope that the present work will foster
further research on non-affine models like the Inverse Gamma stochastic
volatility model, all the more so as this robust model is of great
interest to the industry. 

\noindent \vspace{1mm}

\end{abstract}

\section{Introduction}

\noindent The banking industry, especially in equity and foreign exchange,
is currently experiencing a shift away from affine stochastic volatility
models such as the Heston model, and towards non-affine stochastic
volatility models such as the Inverse Gamma model. Non-affine stochastic
volatility models have been shown to produce more realistic volatility
paths and volatility distributions, to capture more accurately the
dynamics of the market implied volatility surfaces, and to produce
more reliable calibrations, thus reducing the realized volatility
of delta-hedging P\&Ls.

\noindent Up until now, the popularity of affine models in spite of
their empirical inadequacy has been due to one thing: tractability.
Indeed, affine models provide quasi closed-form formulas for vanilla
option prices by transform methods, in contrast to non-affine models.
The purpose of the present work is to resolve this issue by presenting
a fast pricing method for a non-affine stochastic volatility model.
More precisely, we develop a closed-form expansion for the price of
vanilla options under the non-affine Inverse Gamma stochastic volatility
model, defined by the volatility dynamics $dV_{t}=\kappa_{t}\left(\theta_{t}-V_{t}\right)dt+\lambda_{t}V_{t}dB_{t}$.

\noindent Implementing this new closed-form expansion is straightforward,
and pricing speed is instantaneous. In fact, the closed-form expansion
approach is much easier and much faster than the transform methods
used for affine models. Moreover, our method is designed to deal naturally
with time-dependent parameters. This freedom for the term structure
of the model parameters makes the calibration process much easier
for various maturities.

\noindent We illustrate the accuracy of this closed-form expansion
method on several foreign exchange market data sets. The speed and
accuracy of the method make it ideal for industry use. The parameters
generated through the fast calibration procedure can be used to directly
price and hedge options under IGa stochastic volatility, but can also
be used in the calibration of more general local-stochastic volatility
models. 

The paper is organised as follows:
\begin{itemize}
\item Section \ref{sec:IGa_model} defines the Inverse Gamma stochastic
volatility model, discusses similar models in the literature, and
discusses the advantages the Inverse Gamma model has over other classical
one-factor stochastic volatility models.
\item Section \ref{sec:Fast_pricing} provides a closed-form volatility-of-volatility
expansion for the price of a European put option under Inverse Gamma
stochastic volatility. Furthermore, we provide an algorithm to easily
compute the expansion coefficients for piecewise constant parameters,
leading to fast calibrations.
\item Section \ref{sec:Numerical_experiments} provides several numerical
tests of the method on foreign exchange market data (AUD/USD, USD/JPY,
USD/SGD). In each example, the Inverse Gamma model is calibrated to
the whole implied volatility surface and the expansion prices are
compared to Monte Carlo prices. This allows us to assess both calibration
error and expansion error. 
\item Section \ref{sec:Conclusion} summarizes the presented methodology
and provides some of our plans for future work in this area. 
\end{itemize}

\section{The Inverse Gamma Stochastic volatility model\label{sec:IGa_model}}

This section defines the Inverse Gamma stochastic volatility model
and discusses its properties. In this paper we use notations specific
to foreign exchange (namely domestic and foreign interest rates),
but note that the model itself is not limited to foreign exchange
applications and can of course be readily used for other markets (equity,
fixed-income, etc.)

\subsection{Definition}

Denote $S_{t}$ and $V_{t}$ as an exchange rate and its instantaneous
volatility at time $t$, and $T$ the time horizon considered. The
dynamics of the Inverse Gamma (IGa) stochastic volatility model with
time-dependent parameters is given by
\begin{eqnarray}
dS_{t} & = & (r_{d}(t)-r_{f}(t))S_{t}dt+V_{t}S_{t}dW_{t}\nonumber \\
dV_{t} & = & \kappa_{t}\left(\theta_{t}-V_{t}\right)dt+\lambda_{t}V_{t}dB_{t}\label{eq:IGa}\\
d\left\langle W,B\right\rangle _{t} & = & \rho_{t}dt\nonumber 
\end{eqnarray}
where $\left(W_{t},B_{t}\right)_{0\leq t\leq T}$ is a two-dimensional
correlated Brownian motion, $r_{d}=\left(r_{d}(t)\right)_{0\leq t\leq T}$
is the domestic interest rate, and $r_{f}=\left(r_{f}(t)\right)_{0\leq t\leq T}$
is the foreign interest rate. There are four deterministic parameters:
\begin{itemize}
\item $\kappa=\left(\kappa_{t}\right)_{0\leq t\leq T}$ is the rate of mean
reversion of the volatility to the level $\theta$.
\item $\theta=\left(\theta_{t}\right)_{0\leq t\leq T}$ is the mean reversion
level of the volatility.
\item $\lambda=\left(\lambda_{t}\right)_{0\leq t\leq T}$ is the volatility
of volatility.
\item $\rho=\left(\rho_{t}\right)_{0\leq t\leq T}$ is the correlation between
the respective Brownian motions of the underlying $S$ and its volatility
$V$.
\end{itemize}
When the parameters are kept constant, the volatility \eqref{eq:IGa}
is driven by an Inverse Gamma process, yielding an inverse gamma distribution
for the stationary distribution of volatility (cf. Appendix \ref{sub:stationary_IGa}).
Thus, we denote this model as the \textit{Inverse Gamma stochastic
volatility model} (IGa model in short) with time-dependent parameters\footnote{See also Appendix \ref{sec:log_normal_really}.}.

\subsection{The IGa model in the literature\label{sub:litterature}}

A few classes of stochastic volatility models proposed in the literature
contain the IGa model with constant parameters as a particular case.
To make comparisons simpler, we use the same notations for the parameters
of each class ($\kappa$, $\theta$, $\lambda$, $\rho$) and we remove
the drift term from the dynamics of the underlying.
\begin{itemize}
\item The Power Arch (or PARCH) stochastic volatility model (\citet{Fornari01}),
\begin{eqnarray*}
dS_{t} & = & V_{t}S_{t}dW_{t}\\
dV_{t}^{p} & = & \kappa\left(\theta-V_{t}^{p}\right)dt+\lambda V_{t}^{p}dB_{t}
\end{eqnarray*}
corresponds to the IGa stochastic volatility model when $p=1$. Remark
that $p=2$ corresponds to the GARCH diffusion model (cf. Table \ref{tab:One_factor_models}).
\item The Double Log-Normal stochastic volatility model (\citet{Gatheral07,Gatheral08,HenryLabordere09})
with its two cointegrated variance factors:
\begin{eqnarray}
dS_{t} & = & \sqrt{V_{t}}S_{t}dW_{t}\nonumber \\
dV_{t} & = & \kappa\left(V_{t}^{'}-V_{t}\right)dt+\lambda V_{t}dB_{t}\label{eq:DoubleLogNormal}\\
dV_{t}^{'} & = & \kappa^{'}\left(\theta-V_{t}^{'}\right)dt+\lambda^{'}V_{t}^{'}dB_{t}^{'}\nonumber 
\end{eqnarray}
with correlations between the Brownian motions $W$, $B$ and $B^{'}$.
Indeed, the variance formulation of the IGa model can be reformulated
as follows:
\begin{eqnarray*}
dS_{t} & = & \sqrt{V_{t}}S_{t}dW_{t}\\
dV_{t} & = & \left(2\kappa\theta V_{t}^{'}-\left[2\kappa-\lambda^{2}\right]V_{t}\right)dt+2\lambda V_{t}dB_{t}\\
dV_{t}^{'} & = & \kappa\left(\theta-V_{t}^{'}\right)dt+\lambda V_{t}^{'}dB_{t}
\end{eqnarray*}
which is a particular case of Double Log-Normal stochastic volatility
with $100\%$ correlation between $B$ and $B^{'}$.
\item The $\lambda-\mathrm{SABR}$ model\footnote{which, with our notations, is more accurately described as a $\kappa-\mathrm{SABR}$
model.} (\citet{HenryLabordere08} Chapter 6), also known as mean-reverting
SABR:
\begin{eqnarray*}
dS_{t} & = & V_{t}S_{t}^{\beta}dW_{t}\\
dV_{t} & = & \kappa\left(\theta-V_{t}\right)dt+\lambda V_{t}dB_{t}
\end{eqnarray*}
One can see that the IGa model corresponds to the case $\beta=1$,
which, along with the case $\beta=1/2$, is often considered in practice
(see for example \citet{Shiraya11} or \citet{Shiraya14}).
\item The Generalized Inverse Gamma (GIGa) stochastic volatility model of
\citet{Ma14}:
\begin{eqnarray*}
dS_{t} & = & V_{t}S_{t}dW_{t}\\
dV_{t} & = & \kappa\left(\theta V_{t}^{1-\gamma}-V_{t}\right)dt+\lambda V_{t}dB_{t}
\end{eqnarray*}
The special case $\gamma=1$ corresponds to an IGa diffusion for the
volatility.
\item Finally, the closest model to \eqref{eq:IGa} in the literature is
the so-called ``Log-normal Beta stochastic volatility model'' of
\citet{Sepp14,Sepp15}:
\begin{eqnarray}
dS_{t} & = & V_{t}S_{t}dW_{t}\nonumber \\
dV_{t} & = & \kappa\left(\theta-V_{t}\right)dt+\beta V_{t}dW_{t}+\varepsilon V_{t}dB_{t}\label{eq:LN-Beta}\\
d\left\langle W,B\right\rangle _{t} & = & 0\nonumber 
\end{eqnarray}
The model \eqref{eq:LN-Beta} is in fact equivalent to the Inverse
Gamma model \eqref{eq:IGa}, as if $W$ and $B$ are correlated Brownian
motions (with correlation $\rho$), then $B=\rho W+\sqrt{1-\rho^{2}}W^{\bot}$
where $W^{\bot}$ is another Brownian motion, independent from $W$.
Therefore \eqref{eq:LN-Beta} is equivalent to \eqref{eq:IGa} with
$\beta=\lambda\rho$ and $\varepsilon=\lambda\sqrt{1-\rho^{2}}$.
For example, the typical equity case $\beta\approx-1$ and $\varepsilon\approx1$
mentioned in \citet{Sepp14,Sepp15} corresponds to a volatility of
volatility $\lambda=\sqrt{2}\approx1.41$ and a correlation $\rho=-1/\sqrt{2}\approx-0.71$.
\end{itemize}
Other classes of stochastic volatility that contain the IGa model
include $dV_{t}=\kappa\left(\theta-V_{t}\right)dt+\lambda V_{t}^{\eta}dB_{t}$
(\citet{Jerbi11}, IGa when $\eta=1$) and the general $dV_{t}=\left[q\left(t\right)V_{t}^{a}-s\left(t\right)V_{t}^{b}\right]dt+l\left(t\right)V_{t}^{\gamma+1}dB_{t}$
(\citet{Itkin13}, IGa when $q\equiv\kappa\theta$, $s\equiv\kappa$,
$l\equiv\lambda$, $a=0$, $b=1$, though they focus on the closed
forms that can be derived when $a=1$ and $b=2\gamma+1$, which excludes
the IGa model).

This list of models does suggest that the IGa model \eqref{eq:IGa}
is a sensible and reliable basis to model volatility, but also that
the most efficient way to parsimonously enrich the model is not clear
yet. This question is left for future research, and the rest of the
paper will focus on the IGa model \eqref{eq:IGa}.

\subsection{Other models\label{sub:Other-models}}

Over time, many stochastic volatility models have been proposed in
the literature. Table \ref{tab:One_factor_models} recalls some classical
one factor stochastic volatility models with mean-reversion and correlation
between volatility and underlying\footnote{$d\left\langle W,B\right\rangle _{t}=\rho dt$ in all models from
Table \ref{tab:One_factor_models} } with constant parameters. To make comparisons to the IGa model easier,
both volatility and variance formulations are given.\clearpage{}

\begin{table}[h]
\noindent \begin{raggedright}
\hspace{-3mm}%
\begin{tabular}[t]{>{\raggedright}p{21.6mm}|>{\raggedright}p{66.7mm}|>{\raggedright}p{68.1mm}}
Name & Volatility formulation & Variance formulation\tabularnewline[3mm]
 & $dS_{t}\!=(r_{d}-r_{f})S_{t}dt+V_{t}S_{t}dW_{t}$ & $dS_{t}\!=(r_{d}-r_{f})S_{t}dt+\sqrt{V_{t}}S_{t}dW_{t}$\tabularnewline[3mm]
\hline 
\vspace{3mm}
Sch\"obel-Zhu\footnotemark[1] & \vspace{3mm}
$dV_{t}\!=\kappa\left(\theta-V_{t}\right)dt+\lambda dB_{t}$ & \vspace{3mm}
$dV_{t}\!=\left(\lambda^{2}\!+2\kappa\theta\sqrt{V_{t}}-2\kappa V_{t}\right)\!dt+2\lambda\sqrt{V_{t}}dB_{t}$\tabularnewline[5mm]
Heston\footnotemark[2] & $dV_{t}\!=\left(\left[\frac{\kappa\theta}{2}-\frac{\lambda^{2}}{8}\right]\frac{1}{V_{t}}-\frac{\kappa}{2}V_{t}\right)dt+\frac{\lambda}{2}dB_{t}$ & $dV_{t}\!=\kappa\left(\theta-V_{t}\right)dt+\lambda\sqrt{V_{t}}dB_{t}$\tabularnewline[5mm]
3/2-model\footnotemark[3] & $dV_{t}\!=\left(\frac{\kappa\theta}{2}V_{t}-\left[\frac{\kappa}{2}+\!\frac{\lambda^{2}}{8}\right]\!V_{t}^{3}\right)\!dt+\frac{\lambda}{2}V_{t}^{2}dB_{t}$ & $dV_{t}\!=\kappa\left(\theta V_{t}-V_{t}^{2}\right)dt+\lambda V_{t}^{\frac{3}{2}}dB_{t}$\tabularnewline[5mm]
Log-Normal\footnotemark[4]$^{\negthinspace,\negthinspace}$\footnotemark[5] & $dV_{t}\!=\left(\left[\kappa\theta\!+\!\frac{\lambda^{2}}{2}\right]\!V_{t}-\kappa V_{t}\log(V_{t})\!\right)\!dt\!+\!\lambda V_{t}dB_{t}$ & $dV_{t}\!=\left(2\!\left[\kappa\theta\!+\!\lambda^{2}\right]\!V_{t}-\kappa V_{t}\!\log(V_{t})\!\right)\!dt\!+\!2\lambda V_{t}dB_{t}$\tabularnewline[5mm]
GARCH\footnotemark[6] & $dV_{t}\!=\left(\frac{\kappa\theta}{2}\frac{1}{V_{t}}-\left[\frac{\kappa}{2}-\!\frac{\lambda^{2}}{8}\right]\!V_{t}\right)\!dt+\frac{\lambda}{2}V_{t}dB_{t}$ & $dV_{t}\!=\kappa\left(\theta-V_{t}\right)dt+\lambda V_{t}dB_{t}$\tabularnewline[5mm]
\textbf{Inverse Gamma} & $dV_{t}\!=\kappa\left(\theta-V_{t}\right)dt+\lambda V_{t}dB_{t}$ & $dV_{t}\!=\left(2\kappa\theta\sqrt{V_{t}}-\left[2\kappa-\lambda^{2}\right]\!V_{t}\right)\!dt\!+\!2\lambda V_{t}dB_{t}$\tabularnewline[5mm]
\end{tabular}
\par\end{raggedright}

\protect\caption{\label{tab:One_factor_models}One factor stochastic volatility models}

\end{table}

\footnotetext[1]{\citet{Schobel99}}

\footnotetext[2]{\citet{Heston93}}

\footnotetext[3]{\citet{Lewis00}}

\footnotetext[4]{\citet{Wiggins87}}

\footnotetext[5]{A more natural definition is $dS_{t}=(r_{d}-r_{f})S_{t}dt+e^{V_{t}}S_{t}dW_{t}$
with $dV_{t}=\kappa\left(\theta-V_{t}\right)dt+\lambda dB_{t}$. See
also Appendix \ref{sec:log_normal_really}.}

\footnotetext[6]{\citet{Lewis00}}

\footnotetext[7]{using \citet{Rojo96}'s tail classification for
example}

Among the models listed in Table \ref{tab:One_factor_models}, the
models of Sch\"obel-Zhu and Heston are affine, which means that the
Fourier transform of the log-price can be computed explicitly. Because
of their tractability, affine models have received a lot of attention
in the literature, at the expense of the non-affine stochastic volatility
models. Unfortunately, empirical analyzes suggest that the dynamics
of market volatilities is much better described by non-affine models.
Let us illustrate this point by comparing the Heston model and the
Inverse Gamma model.

Figure \ref{fig:vol_dens} displays the stationary distribution of
the volatility under the Heston and Inverse Gamma models, with same
mean ($0.30$) and same standard deviation ($0.08$ on Figures \ref{fig:vol_std_0.08}
and \ref{fig:vol_std_0.08_log}, $0.16$ on Figures \ref{fig:vol_std_0.16}
and \ref{fig:vol_std_0.16_log}, $0.24$ on Figures \ref{fig:vol_std_0.24}
and \ref{fig:vol_std_0.24_log}), using the results from Appendix
\ref{sec:stationary_vol}.

On the one hand, the volatility distribution in the affine models
from Table \ref{tab:One_factor_models} (Sch\"obel-Zhu, Heston) has
a short right tail\footnotemark[7], while it has a more realistic
long right tail in the non-affine models (3/2 model, Log-Normal, GARCH,
and Inverse Gamma). Figures \ref{fig:vol_std_0.08_log}, \ref{fig:vol_std_0.16_log}
and \ref{fig:vol_std_0.24_log} (in log-scale), illustrate this difference
between Heston and Inverse Gamma. The right tail of the Heston volatility
decreases more quickly than that of the IGa volatility, therefore
there is always a volatility level upon which the Heston volatility
falls forever below the IGa one.

The left tail is also better described with non-affine models. For
example, with the Heston model, the volatility can reach zero if the
parameters ($\kappa$, $\theta$, $\lambda$, $\rho$) do not satisfy
the Feller condition ($2\kappa\theta/\lambda^{2}>1$). The effect
of this condition on the left tail can be seen from Figures \ref{fig:vol_std_0.08}
and \ref{fig:vol_std_0.08_log} ($2\kappa\theta/\lambda^{2}=3.63$)
to Figures \ref{fig:vol_std_0.16} and \ref{fig:vol_std_0.16_log}
($2\kappa\theta/\lambda^{2}=0.96$) to Figures \ref{fig:vol_std_0.24}
and \ref{fig:vol_std_0.24_log} ($2\kappa\theta/\lambda^{2}=0.49$).
One can clearly see how the distribution piles up close to zero, to
the point where zero becomes the most likely value for the volatility
(Figure \ref{fig:vol_std_0.24}). Unfortunately, the Feller condition
is almost always violated in practice (\citet{Clark11}, \citet{Fonseca11},
\citet{Ribeiro13}, $\ldots$), which means that Figure \ref{fig:vol_std_0.24}
represents the normal behavior of the Heston model on real data.

\hspace{-3mm}
\begin{figure}[!th]
\subfloat[\label{fig:vol_std_0.08}mean $=0.30$, standard deviation $=0.08$]{\protect\includegraphics[width=0.38\paperwidth]{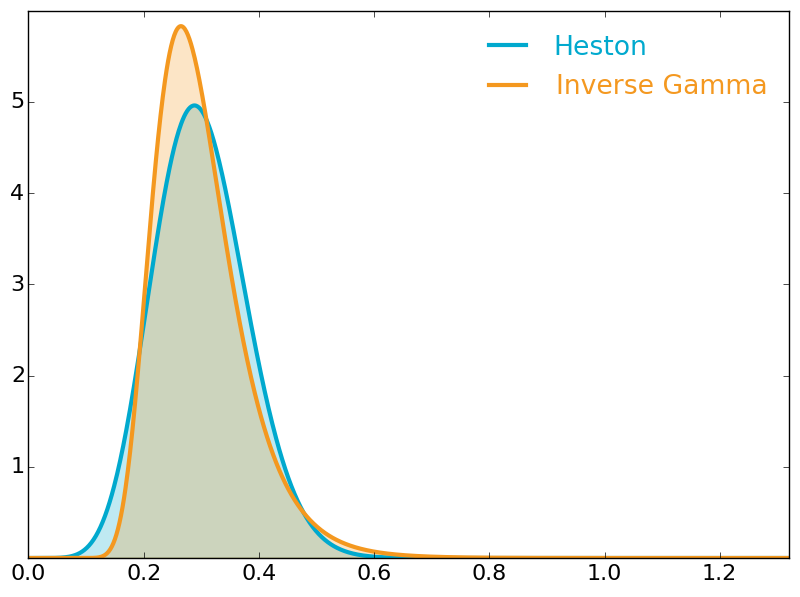}

}\subfloat[\label{fig:vol_std_0.08_log}mean $=0.30$, standard deviation $=0.08$,
log scale]{\protect\includegraphics[width=0.38\paperwidth]{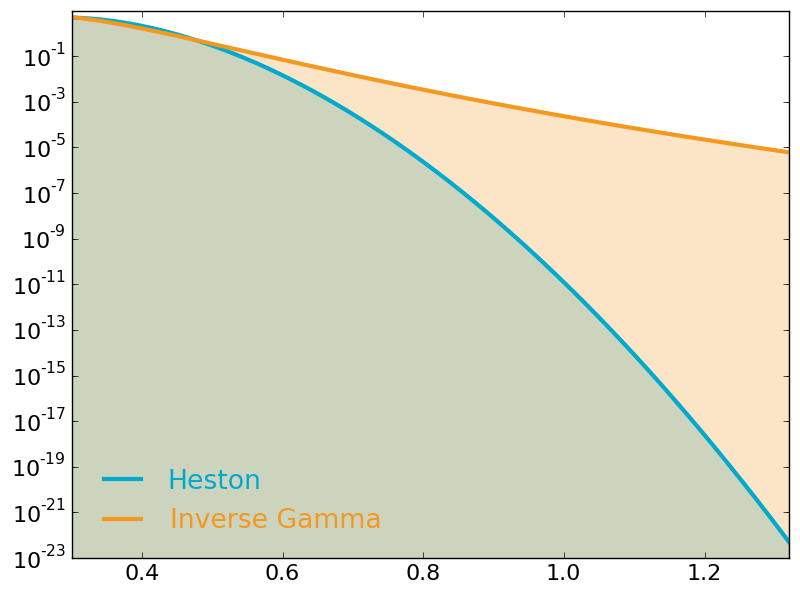}

}\\
\subfloat[\label{fig:vol_std_0.16}mean $=0.30$, standard deviation $=0\lyxdot{} 16$]{\protect\includegraphics[width=0.38\paperwidth]{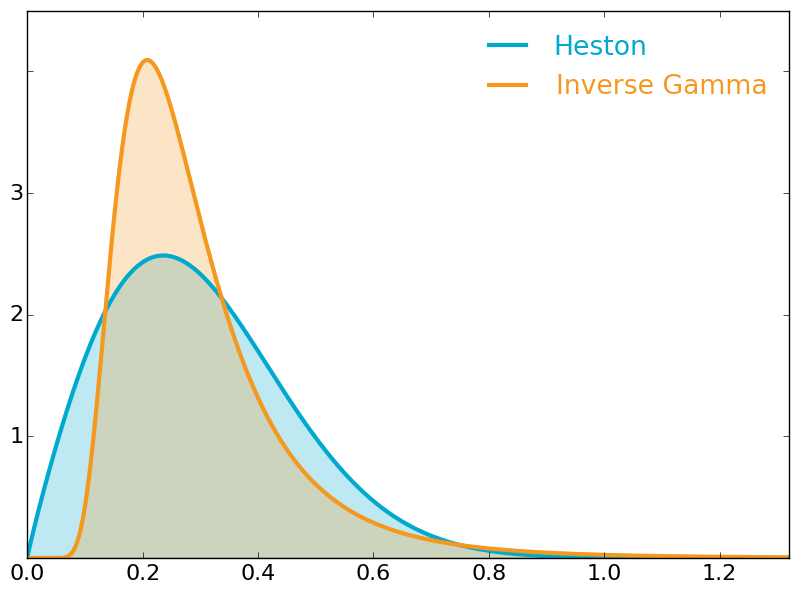}

}\subfloat[\label{fig:vol_std_0.16_log}mean $=0.30$, standard deviation $=0.16$,
log scale]{\protect\includegraphics[width=0.38\paperwidth]{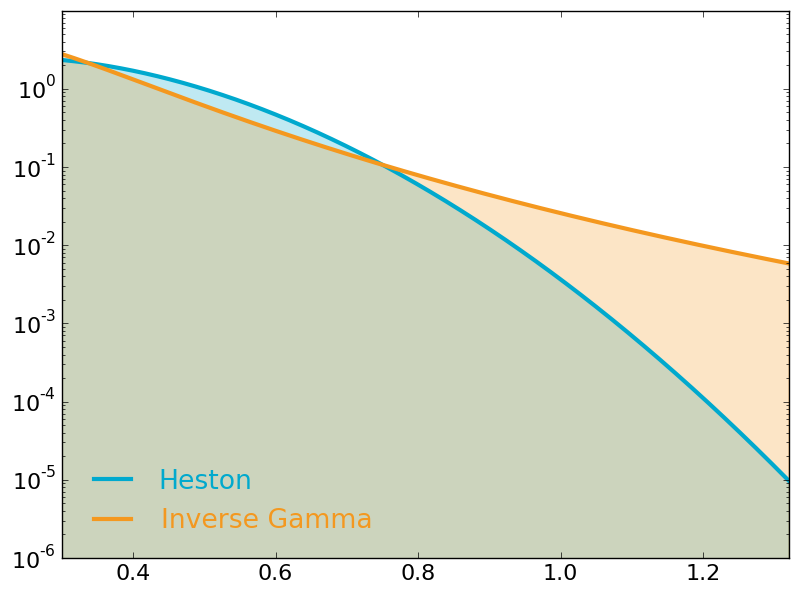}

}\\
\subfloat[\label{fig:vol_std_0.24}mean $=0.30$, standard deviation $=0.24$]{\protect\includegraphics[width=0.38\paperwidth]{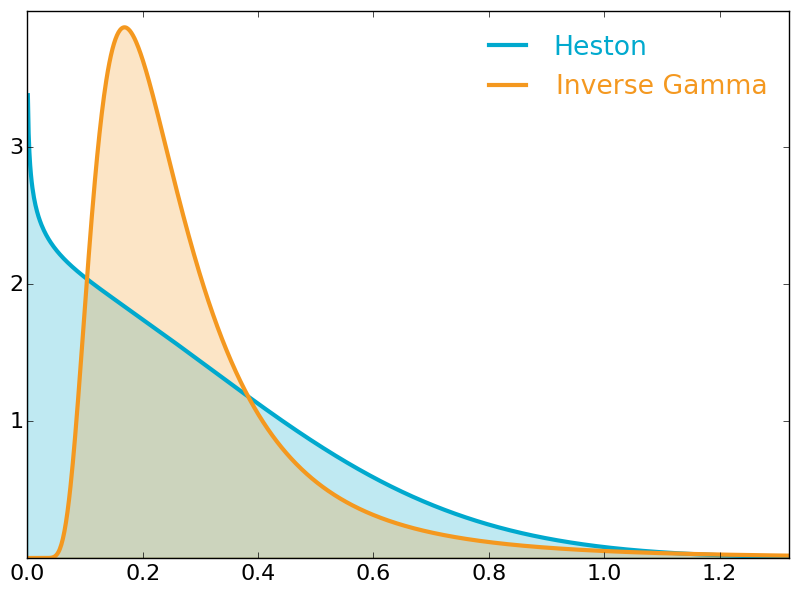}

}\subfloat[\label{fig:vol_std_0.24_log}mean $=0.30$, standard deviation $=0.24$,
log scale]{\protect\includegraphics[width=0.38\paperwidth]{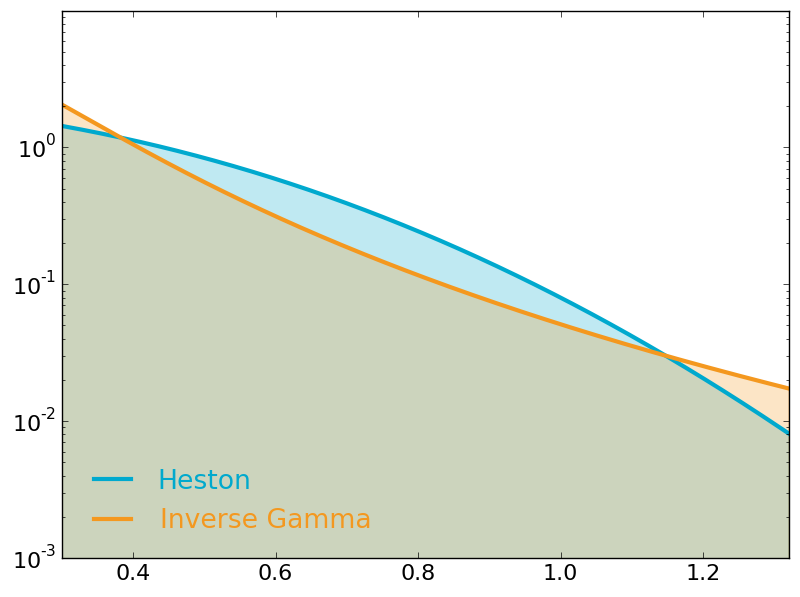}

}

\protect\caption{\label{fig:vol_dens}Volatility density}

\end{figure}

\subsubsection{Empirical evidence}

Figure \ref{fig:vol_dens} strongly suggests that the IGa model \ref{eq:IGa},
while being as parsimonious as Heston (four parameters $\kappa$,
$\theta$, $\lambda$, $\rho$), provides a much better description
of realized and implied market volatility. Indeed, a number of empirical
studies suggest that a non-affine stochastic volatility model of the
type \eqref{eq:IGa} compares favorably to other possible models,
especially to affine models such as the Heston model. We can summarize
these studies as listed below.

In equity markets:
\begin{itemize}
\item In \citet{Bouchaud03} (Chapter 7), an empirical analysis of the volatility
of the S\&P 500 in the period 1990-2001 is performed. The distribution
of the volatility is shown to be accurately fitted by two distributions:
a log-normal distribution, and an inverse gamma distribution. Overall,
the best fit is provided by the inverse gamma distribution, especially
on the right tail of the volatility distribution (Figures 7.7 p.118
and 7.8 p.119). When discussing the shortcomings of the Heston stochastic
volatility models, the authors explicitly say that ``the empirical
distribution of the volatility is closer to an \textit{inverse} gamma
distribution than to a gamma distribution'' (p.143).
\item \citet{Gander07} test several possible distributions for the volatility
(Tempered stable, Generalized Inverse Gaussian (including Gamma),
Positive Hyperbolic, Inverse Gaussian and Inverse Gamma) on $14$
stocks traded on the New York Stock Exchange. It is shown that the
Inverse Gamma distribution provides the best fit for option pricing.
\item \citet{Gatheral08} shows that the Double Lognormal stochastic volatility
model (which contains \eqref{eq:IGa} as a special case) fits SPX
and VIX options much better than Double Heston, with stable parameters.
\item \citet{Christoffersen10} show, on S\&P 500 returns, VIX options and
OTM S\&P 500 index option data, between $1996$ and $2004$, that
the GARCH diffusion model (Inverse Gamma variance) significantly outperforms
several other models including the Heston and the $3/2$ models. (They
only consider the volatility parameterization (not variance) $dV_{t}=\left(\frac{\kappa\theta}{2}V_{t}^{2a-1}-\frac{\kappa}{2}V_{t}^{2a+1}-\frac{\lambda^{2}}{8}V_{t}^{4b-3}\right)dt+\frac{\lambda}{2}V_{t}^{2b-1}dB_{t}$,
which does not contain the Inverse Gamma model.)
\item Furthermore, \citet{Kaeck12} show that allowing for non-affine dynamics
(like \eqref{eq:IGa}) is more important than the inclusion of jumps.
In particular, augmenting an affine model such as the Heston model
with jumps leads to a stochastic volatility model that is still significantly
inferior to a more parsimonious non-affine model like GARCH without
jumps.
\item Finally, \citet{Ma14} analyze the volatility of the S\&P 100, S\&P
500 and DJIA, as well as the VIX index between $1990$ and $2014$.
They show that the Generalized Inverse Gamma distribution (which contain
the Inverse Gamma distribution) fits volatility best.
\end{itemize}
In fixed-income markets:
\begin{itemize}
\item \citet{Fornari01} calibrate the Power Arch model to futures contracts
on the Italian 10-year government bond, between 1991 to 1997. The
power $p$ is estimated on three different subsamples. The results,
$0.86$, $0.99$ and $1.19$, are very close to an IGa volatility
($p=1$). In particular, the GARCH diffusion model ($p=2$) is rejected. 
\item \citet{Fornari06} then study a slightly different stochastic volatility
model, of the type $dV_{t}^{p}=\kappa(\theta-V_{t}^{p})dt+\lambda V_{t}^{\eta p}dB_{t}$,
with two parameters $p$ and $\eta$. They fit these two parameters
to weekly $3$-month US Treasury bills rates between 1973 and 1995.
Their estimates, $\hat{p}=1.0326$ and $\hat{\eta}=1.0014$, are statistically
indistinguishable from $1$ (Inverse Gamma volatility). Remark that
the volatility term for the asset $S$ (which is an interest rate
in there case) is of the form $V_{t}\sqrt{S_{t}}dW_{t}$, which is
different from \eqref{eq:IGa}. However, they state that a volatility
of the form $V_{t}\left|S_{t}\right|^{d}dW_{t}$, $d\geq0.5$ (which
contains \eqref{eq:IGa} for $d=1$), would not dramatically change
their empirical results.
\end{itemize}
Beyond better calibration and more realistic volatility distribution
and volatility paths, a model such as \ref{eq:IGa} generates a more
realistic dynamics for the implied volatility surface (\citet{Tataru12}),
which reduces the tracking volatility in hedging P\&Ls (\citet{Sepp15}),
and drastically improves portfolio allocation (\citet{Hansis10}).

\section{Closed-form expansion for fast option pricing\label{sec:Fast_pricing}}

In the previous section, a number of reasons have been documented
in favor of a non-affine Inverse Gamma stochastic volatility model
over affine models. However, the main reason why affine models are
used in practice is not their realism but their tractability. Indeed
their Fourier (or Laplace) transform is available in closed form,
which makes pricing possible by inverse transform. Non-affine models
such as the IGa model do not have a closed form solution for their
Fourier transform, which makes them \textit{a priori} less tractable.

In \citet{Sepp14}, an affine moment-matching approximation of the
moment generation function for \eqref{eq:IGa} is proposed, making
pricing possible by inverse transform. The advantage of this approach
is that jumps can be factored in during the matching as well. Here
we propose a more straightforward approach, namely a closed-form volatility-of-volatility
expansion for vanilla options prices. Compared to \citet{Sepp14},
the main advantages of our approach are that:
\begin{itemize}
\item It is simpler. If an approximation is to be made, it is more straightforward
and intuitive to approximate the price directly rather than a transform
of the price. This also makes pricing and subsequent calibration much
faster.
\item No moment matching is required. Firstly, matching moments with an
affine model can create unexpected problems (recall Figure \ref{fig:vol_std_0.24}).
Secondly, for most stochastic volatility models, moments higher than
$1$ cease to exist for large maturities (\citet{Andersen07}). In
\citet{Sepp14}, the author advises against going beyond a second-order
approximation, as only $5$ moments were shown to exist. Our approach
does not have such limitation on the order of approximation.
\item Importantly, our approximation approach is naturally suitable for
time-dependent parameters.
\end{itemize}
In the following subsection \ref{sub:closed_form_expansion}, we provide
a closed-form expansion for the price of a European put option under
Inverse Gamma volatility \eqref{eq:IGa} with time-dependent parameters
(equation \eqref{eq:IGa_put_expansion}). The method is based on the
methodology developed in \citet{Benhamou10}, adapted to the Inverse
Gamma model and extended to time-dependent $\kappa$. In \citet{Benhamou10},
the closed-form volatility of volatility expansion methodology was
applied to the Heston stochastic volatility with time-dependent parameters.
It was shown to be very accurate, and much faster than Fourier methods. 

The coefficients of the closed-form expansion are explicitly given
by time integrals of the (time-dependent) parameters (equation \eqref{eq:coefficients}).
Importantly, any shape for the time evolution of the parameters can
be handled. In practice though, piecewise-constant parameters can
be seen as a good compromise between richness and tractability. Thus,
we present generic recusions formulas for the coefficients \eqref{eq:coefficients}
when the model parameters are piecewise constant. These generic recursions
are easy to implement, and can be used for expansion coefficients
to any order.

Finally, in order to compare our prices to another method, and as
transform methods are not available, we explain how to implement an
efficient Monte Carlo scheme to price options under IGa volatility.
\newpage{}

\subsection{Closed-form expansion\label{sub:closed_form_expansion}}
\begin{thm}
\label{thm:IGa_expansion}The second-order expansion for the price
$P_{IGa}=P_{IGa}\left(S_{0},K,T,r_{d},r_{f};\kappa,\theta,\lambda,\rho\right)$
of a European put option with Inverse Gamma volatility is explicitly
given by{\large{}
\begin{equation}
P_{IGa}=P_{BS}\left(x_{0},\psi_{T}\right)+\sum_{i=0}^{2}a_{i,T}\frac{\partial^{i+1}}{\partial x^{i}y}P_{BS}\left(x_{0},\psi_{T}\right)+\sum_{i=0}^{1}b_{2i,T}\frac{\partial^{2i+2}}{\partial x^{2i}y^{2}}P_{BS}\left(x_{0},\psi_{T}\right)+\mathcal{E}\label{eq:IGa_put_expansion}
\end{equation}
}{\large \par}

with{\large{}
\begin{eqnarray}
x_{0} & = & \log\left(S_{0}\right)\nonumber \\
v_{0,t} & = & e^{-\int_{0}^{t}\kappa_{z}dz}\left(v_{0}+\int_{0}^{t}\kappa_{s}\theta_{s}e^{\int_{0}^{s}\kappa_{z}dz}ds\right)\\
\psi_{T} & = & \int_{0}^{T}v_{0,t}^{2}dt\\
a_{0,T} & = & \int_{0}^{T}e^{\int_{0}^{s}2\kappa_{z}dz}\lambda_{s}^{2}v_{0,s}^{2}ds\int_{s}^{T}e^{-\int_{0}^{t}2\kappa_{z}dz}dt\\
a_{1,T} & = & 2\int_{0}^{T}e^{\int_{0}^{s}\kappa_{z}dz}\rho_{s}\lambda_{s}v_{0,s}^{2}ds\int_{s}^{T}e^{-\int_{0}^{t}\kappa_{z}dz}v_{0,t}dt\nonumber \\
a_{2,T} & = & 2\int_{0}^{T}e^{\int_{0}^{s}\kappa_{z}dz}\rho_{s}\lambda_{s}v_{0,s}^{2}ds\int_{s}^{T}2\rho_{t}\lambda_{t}v_{0,t}dt\int_{t}^{T}e^{-\int_{0}^{u}\kappa_{z}dz}v_{0,u}du\nonumber \\
 & + & 2\int_{0}^{T}e^{\int_{0}^{s}\kappa_{z}dz}\rho_{s}\lambda_{s}v_{0,s}^{2}ds\int_{s}^{T}e^{\int_{0}^{t}\kappa_{z}dz}\rho_{t}\lambda_{t}v_{0,t}^{2}dt\int_{t}^{T}e^{-\int_{0}^{u}2\kappa_{z}dz}du\nonumber \\
b_{0,T} & = & 4\int_{0}^{T}e^{\int_{0}^{s}2\kappa_{z}dz}\lambda_{s}^{2}v_{0,s}^{2}ds\int_{s}^{T}e^{-\int_{0}^{t}\kappa_{z}dz}v_{0,t}dt\int_{t}^{T}e^{-\int_{0}^{u}\kappa_{z}dz}v_{0,u}du\nonumber \\
b_{2,T} & = & \frac{a_{1,T}^{2}}{2}\label{eq:coefficients}
\end{eqnarray}
}{\large \par}

where $P_{BS}\left(x,y\right)=P_{BS}\left(x,y;K,T,r_{d},r_{f}\right)$
is the Black-Scholes put price with spot $e^{x}$ and integrated variance
$y$,
\begin{eqnarray}
P_{BS}\left(x,y\right) & = & Ke^{-\int_{0}^{T}r_{d}\left(t\right)dt}\mathcal{N}\left(\frac{1}{\sqrt{y}}\log\left(\frac{Ke^{-\int_{0}^{T}r_{d}\left(t\right)dt}}{e^{x}e^{-\int_{0}^{T}r_{q}\left(t\right)dt}}\right)+\frac{1}{2}\sqrt{y}\right)\nonumber \\
 &  & -e^{x}e^{-\int_{0}^{T}r_{q}\left(t\right)dt}\mathcal{N}\left(\frac{1}{\sqrt{y}}\log\left(\frac{Ke^{-\int_{0}^{T}r_{d}\left(t\right)dt}}{e^{x}e^{-\int_{0}^{T}r_{q}\left(t\right)dt}}\right)-\frac{1}{2}\sqrt{y}\right)\,,\label{eq:PBS}
\end{eqnarray}
and $\mathcal{E}$ is the error term of the second-order expansion.\end{thm}
\begin{proof}
The proof or expansion \eqref{eq:IGa_put_expansion} is available
in Appendix \ref{sec:proof}. It is based on the proof of the Heston
expansion in \citet{Benhamou10}, that we extended to non-constant
$\kappa$, and adapted to the IGa model. \end{proof}
\begin{rem}
One can also easily obtain closed-form expansion for the Greeks by
adapting the proof of Theorem \ref{thm:IGa_expansion} to it. More
generally the methodology adopted in Theorem \ref{thm:IGa_expansion}
can be adapted to any options that have closed-form Black-Scholes
prices with time dependent parameters, for example barrier options
(\citet{Lo03}, \citet{Rapisarda03}).
\end{rem}

\subsection{Recursions for expansion coefficients\label{sub:coefficients_recursions}}

The above coefficients \eqref{eq:coefficients} are expressed for
general deterministic parameters $\left(\kappa_{t},\theta_{t},\lambda_{t},\rho_{t}\right)_{0\leq t\leq T}$.
In this subsection, we propose an explicit recursive algorithm to
compute the coefficients when the parameters are piecewise constant.

Let $T_{0}=0<T_{1}<T_{2}<\cdots<T_{N}=T$ be a partition of $\left[0,T\right]$.
We now suppose that the parameters are constant on each interval:
\begin{eqnarray}
\left(\kappa_{t},\theta_{t},\lambda_{t},\rho_{t}\right) & :=\left(\kappa_{i},\theta_{i},\lambda_{i},\rho_{i}\right) & \forall t\in\left[T_{i},T_{i+1}\right[\label{eq:Piecewise_Constant}
\end{eqnarray}

\subsubsection{Integral operator}

Firstly, we define recursively the following integral operator \footnote{which corresponds to Def. 5.1 in \citet{Benhamou10} extended to non-constant
$\kappa$}
\begin{eqnarray}
\omega_{t,T}^{\left(\kappa,l\right)}=\int_{t}^{T}e^{\int_{0}^{u}\kappa_{z}dz}l_{u}du &  & \forall t\in\left[0,T\right]\label{eq:omega_0}
\end{eqnarray}

\vspace{-4mm}

\begin{eqnarray}
\omega_{t,T}^{\left(\kappa_{n},l_{n}\right),\ldots,\left(\kappa_{1},l_{1}\right)}=\omega_{t,T}^{\left(\kappa_{n},l_{n}\omega_{.,T}^{\left(\kappa_{n-1},l_{n-1}\right),\ldots,\left(\kappa_{1},l_{1}\right)}\right)} &  & \forall t\in\left[0,T\right]\label{eq:omega_n}
\end{eqnarray}
Using this notation, the coefficients from the price expansion \ref{eq:IGa_put_expansion}
can be expressed as follows:
\begin{eqnarray}
\psi_{T} & = & \omega_{0,T}^{\left(0,v_{0,.}^{2}\right)}\nonumber \\
a_{0,T} & = & \omega_{0,T}^{\left(2\kappa,\lambda^{2}v_{0,.}^{2}\right),\left(-2\kappa,1\right)}\nonumber \\
a_{1,T} & = & 2\omega_{0,T}^{\left(\kappa,\rho\lambda v_{0,.}^{2}\right),\left(-\kappa,v_{0,.}\right)}\nonumber \\
a_{2,T} & = & 2\omega_{0,T}^{\left(\kappa,\rho\lambda v_{0,.}^{2}\right),\left(0,2\rho\lambda v_{0,.}\right),\left(-\kappa,v_{0,.}\right)}+2\omega_{0,T}^{\left(\kappa,\rho\lambda v_{0,.}^{2}\right),\left(\kappa,\rho\lambda v_{0,.}^{2}\right),\left(-2\kappa,1\right)}\nonumber \\
b_{0,T} & = & 4\omega_{0,T}^{\left(2\kappa,\lambda^{2}v_{0,.}^{2}\right),\left(-\kappa,v_{0,.}\right),\left(-\kappa,v_{0,.}\right)}\label{eq:Integral_Coefficients}
\end{eqnarray}

\subsubsection{Recursions}

Let $l_{1},l_{2},l_{3},\ldots$ be deterministic functions. When these
functions are piecewise constant,

\[
l_{k}(t):=l_{k,i}\,,\,t\in\left[T_{i},T_{i+1}\right[\,,\,k=1,2,\ldots\,,
\]

then the following integral operators at time $T_{i+1}$ can be expressed
as a function of the integral operators at time $T_{i}$:{\large{}
\begin{eqnarray}
\omega_{0,T_{i+1}}^{\left(n_{1}\kappa,l_{1}v_{0,.}^{p_{1}}\right)} & = & \omega_{0,T_{i}}^{\left(n_{1}\kappa,l_{1}v_{0,.}^{p_{1}}\right)}+e_{0,T_{i}}^{n_{1}}l_{1,i}\varphi_{T_{i},T_{i+1}}^{\left(n_{1},0,p_{1}\right)}\label{eq:Integral_Recursion_1}\\
\nonumber \\
\omega_{0,T_{i+1}}^{\left(n_{2}\kappa,l_{2}v_{0,.}^{p_{2}}\right),\left(n_{1}\kappa,l_{1}v_{0,.}^{p_{1}}\right)} & = & \omega_{0,T_{i}}^{\left(n_{2}\kappa,l_{2}v_{0,.}^{p_{2}}\right),\left(n_{1}\kappa,l_{1}v_{0,.}^{p_{1}}\right)}\nonumber \\
 & + & \omega_{0,T_{i}}^{\left(n_{2}\kappa,l_{2}v_{0,.}^{p_{2}}\right)}e_{0,T_{i}}^{n_{1}}l_{1,i}\varphi_{T_{i},T_{i+1}}^{\left(n_{1},0,p_{1}\right)}\nonumber \\
 & + & e_{0,T_{i}}^{n_{2}+n_{1}}l_{2,i}l_{1,i}\varphi_{T_{i},T_{i+1}}^{\left(n_{2},0,p_{2}\right),\left(n_{1},0,p_{1}\right)}\label{eq:Integral_Recursion_2}
\end{eqnarray}
}\newpage{}

{\large{}
\begin{eqnarray}
\omega_{0,T_{i+1}}^{\left(n_{3}\kappa,l_{3}v_{0,.}^{p_{3}}\right),\left(n_{2}\kappa,l_{2}v_{0,.}^{p_{2}}\right),\left(n_{1}\kappa,l_{1}v_{0,.}^{p_{1}}\right)} & = & \omega_{0,T_{i}}^{\left(n_{3}\kappa,l_{3}v_{0,.}^{p_{3}}\right),\left(n_{2}\kappa,l_{2}v_{0,.}^{p_{2}}\right),\left(n_{1}\kappa,l_{1}v_{0,.}^{p_{1}}\right)}\nonumber \\
 & + & \omega_{0,T_{i}}^{\left(n_{3}\kappa,l_{3}v_{0,.}^{p_{3}}\right),\left(n_{2}\kappa,l_{2}v_{0,.}^{p_{2}}\right)}e_{0,T_{i}}^{n_{1}}l_{1,i}\varphi_{T_{i},T_{i+1}}^{\left(n_{1},0,p_{1}\right)}\nonumber \\
 & + & \omega_{0,T_{i}}^{\left(n_{3}\kappa,l_{3}v_{0,.}^{p_{3}}\right)}e_{0,T_{i}}^{n_{2}+n_{1}}l_{2,i}l_{1,i}\varphi_{T_{i},T_{i+1}}^{\left(n_{2},0,p_{2}\right),\left(n_{1},0,p_{1}\right)}\nonumber \\
 & + & e_{0,T_{i}}^{n_{3}+n_{2}+n_{1}}l_{3,i}l_{2,i}l_{1,i}\varphi_{T_{i},T_{i+1}}^{\left(n_{3},0,p_{3}\right),\left(n_{2},0,p_{2}\right),\left(n_{1},0,p_{1}\right)}\label{eq:Integral_Recursion_3}
\end{eqnarray}
}{\large \par}

and so on, where $n_{1},n_{2},n_{3},\ldots$ and $p_{1},p_{2},p_{3},\ldots$
are integers, $e_{0,t}=e^{\int_{0}^{t}\kappa_{z}dz}$, and for all
$T_{i}\leq t\leq T_{i+1}$, $\varphi$ is defined as follows:

{\large{}
\begin{eqnarray*}
\varphi_{t,T_{i+1}}^{\left(n_{1},m_{1},p_{1}\right)} & = & \int_{t}^{T_{i+1}}e^{n_{1}\int_{T_{i}}^{s}\kappa_{z}dz}\gamma\left(s\right)^{m_{1}}v_{0,s}^{p_{1}}ds=\int_{t}^{T_{i+1}}e^{n_{1}\kappa_{i}\Delta T_{i}\gamma\left(s\right)}\gamma\left(s\right)^{m_{1}}v_{0,s}^{p_{1}}ds\\
\varphi_{t,T_{i+1}}^{\left(n_{k},m_{k},p_{k}\right),\cdots,\left(n_{1},m_{1},p_{1}\right)} & = & \int_{t}^{T_{i+1}}e^{n_{k}\int_{T_{i}}^{s}\kappa_{z}dz}\gamma\left(s\right)^{m_{k}}v_{0,s}^{p_{k}}\varphi_{t,T_{i+1}}^{\left(n_{k-1},m_{k-1},p_{k-1}\right),\cdots,\left(n_{1},m_{1},p_{1}\right)}ds\\
 & = & \int_{t}^{T_{i+1}}e^{n_{k}\kappa_{i}\Delta T_{i}\gamma\left(s\right)}\gamma\left(s\right)^{m_{k}}v_{0,s}^{p_{k}}\varphi_{t,T_{i+1}}^{\left(n_{k-1},m_{k-1},p_{k-1}\right),\cdots,\left(n_{1},m_{1},p_{1}\right)}ds
\end{eqnarray*}
}{\large \par}

where $\gamma\left(s\right)=\frac{s-T_{i}}{\Delta T_{i}}$ with $\Delta T_{i}=T_{i+1}-T_{i}$,
and $n_{1},m_{1},p_{1},\ldots,n_{k},m_{k},p_{k}$ are integers.

It should be noted that when the parameters are piecewise constant
(equation \eqref{eq:Piecewise_Constant}), $\varphi$ can be computed
explicitly by recursion. Define $\Delta\kappa_{i}=\kappa_{i+1}-\kappa_{i}$,
$\Delta\theta_{i}=\theta_{i+1}-\theta_{i}$, $\Delta\lambda_{i}=\lambda_{i+1}-\lambda_{i}$
and $\Delta\rho_{i}=\rho_{i+1}-\rho_{i}$, and let $t\in\left[T_{i},T_{i+1}\right]$.
Then

{\large{}
\[
v_{0,t}=\theta_{i}+\left(v_{0,T_{i}}-\theta_{i}\right)e^{-\kappa_{i}\Delta T_{i}\gamma\left(t\right)}\,,
\]
}{\large \par}

and, using the definition of $\varphi$, basic integration and integration
by parts, the following recursions hold

{\large{}
\[
\varphi_{t,T_{i+1}}^{\left(n_{1},m_{1},p_{1}\right)}=\begin{cases}
\theta_{i}\varphi_{t,T_{i+1}}^{\left(n_{1},m_{1},p_{1}-1\right)}+\left(v_{0,T_{i}}-\theta_{i}\right)\varphi_{t,T_{i+1}}^{\left(n_{1}-1,m_{1},p_{1}-1\right)} & p_{1}>0\\
\\
\frac{\Delta T_{i}}{m_{1}+1}\left(1-\gamma\left(t\right)^{m_{1}+1}\right) & n_{1}=0,\,p_{1}=0\\
\\
\frac{e^{n_{1}\kappa_{i}\Delta T_{i}}-e^{n_{1}\kappa_{i}\Delta T_{i}\gamma\left(t\right)}}{n_{1}\kappa_{i}} & n_{1}\neq0,\,m_{1}=0,\,p_{1}=0\\
\\
\frac{e^{n_{1}\kappa_{i}\Delta T_{i}}-\gamma\left(t\right)^{m_{1}}e^{n_{1}\kappa_{i}\Delta T_{i}\gamma\left(t\right)}}{n_{1}\kappa_{i}}-\frac{m_{1}}{n_{1}\kappa_{i}\Delta T_{i}}\varphi_{t,T_{i+1}}^{\left(n_{1},m_{1}-1,0\right)} & n_{1}\neq0,\,m_{1}>0,\,p_{1}=0
\end{cases}
\]
}{\large \par}

and for every integer $k>1$:

{\large{}
\begin{eqnarray*}
 &  & \varphi_{t,T_{i+1}}^{\left(n_{k},m_{k},p_{k}\right),\left(n_{k-1},m_{k-1},p_{k-1}\right),\cdots,\left(n_{1},m_{1},p_{1}\right)}=\\
\\
 &  & \begin{cases}
\theta_{i}\varphi_{t,T_{i+1}}^{\left(n_{k},m_{k},p_{k}-1\right),\cdots}+\left(v_{0,T_{i}}-\theta_{i}\right)\varphi_{t,T_{i+1}}^{\left(n_{k}-1,m_{k},p_{k}-1\right)} & p_{k}>0\\
\\
-\frac{\Delta T_{i}}{m_{k}+1}\gamma\left(t\right)^{m_{k}+1}\varphi_{t,T_{i+1}}^{\left(n_{k-1},m_{k-1},p_{k-1}\right),\cdots}+\frac{\Delta T_{i}}{m_{k}+1}\varphi_{t,T_{i+1}}^{\left(n_{k-1},m_{k}+m_{k-1}+1,p_{k-1}\right),\cdots} & n_{k}=0,\,p_{k}=0\\
\\
-\frac{e^{n_{k}\kappa_{i}\Delta T_{i}\gamma\left(t\right)}}{n_{k}\kappa_{i}}\varphi_{t,T_{i+1}}^{\left(n_{k-1},m_{k-1},p_{k-1}\right),\cdots}+\frac{1}{n_{k}\kappa_{i}}\varphi_{t,T_{i+1}}^{\left(n_{k}+n_{k-1},m_{k-1},p_{k-1}\right),\cdots} & n_{k}\neq0,\,m_{k}=0,\,p_{k}=0\\
\\
-\frac{e^{n_{k}\kappa_{i}\Delta T_{i}\gamma\left(t\right)}}{n_{k}\kappa_{i}}\left[\sum_{j=0}^{m_{k}}\gamma\left(t\right)^{j}\frac{m_{k}!}{j!}\left(\frac{-1}{n_{k}\kappa_{i}\Delta T_{i}}\right)^{m_{k}-j}\right]\varphi_{t,T_{i+1}}^{\left(n_{k-1},m_{k-1},p_{k-1}\right),\cdots}\\
+\frac{1}{n_{k}\kappa_{i}}\sum_{j=0}^{m_{k}}\frac{m_{k}!}{j!}\left(\frac{-1}{n_{k}\kappa_{i}\Delta T_{i}}\right)^{m_{k}-j}\varphi_{t,T_{i+1}}^{\left(n_{k}+n_{k-1},m_{k-1}+j,p_{k-1}\right),\cdots} & n_{k}\neq0,\,m_{k}>0,\,p_{k}=0
\end{cases}
\end{eqnarray*}
}{\large \par}

All these equation are sufficient to compute $\varphi_{T_{i},T_{i+1}}^{\left(n_{k},m_{k},p_{k}\right),\cdots,\left(n_{1},m_{1},p_{1}\right)}$
for any integers $\left(n_{1},m_{1},p_{1}\right),\cdots,\left(n_{k},m_{k},p_{k}\right)$,
making it possible to implement the integral recursions \eqref{eq:Integral_Recursion_1},
\eqref{eq:Integral_Recursion_2} and \eqref{eq:Integral_Recursion_3},
which, in turn, make it possible to compute explicitly the expansion
coefficients \eqref{eq:Integral_Coefficients} when the parameters
of the stochastic volatility model are piecewise constant.

\subsection{Monte Carlo\label{sub:Monte_Carlo}}

An obvious alternative method for computing the price of a European
put option under the IGa stochastic volatility model is the Monte
Carlo method. Using the expression

\[
P_{IGa}\left(x_{0},v_{0}\right)=\mathbb{E}\left[P_{BS}\left(x_{0}+\int_{0}^{T}\rho_{t}V_{t}dB_{t}-\frac{1}{2}\int_{0}^{T}\left(\rho_{t}V_{t}\right)^{2}dt,\int_{0}^{T}\left(1-\rho_{t}^{2}\right)V_{t}^{2}dt\right)\right]\,,
\]
(cf. equation \eqref{eq:g_PBS}), only the simulation of the volatility
is needed. Then, one can take advantage of the strong solution of
the Inverse Gamma diffusion (\citet{Zhao09})

\[
V_{t}=\frac{1}{Z_{t}}\left(V_{0}+\int_{0}^{t}\kappa_{s}\theta_{s}Z_{s}ds\right)
\]
where $Z$ is a geometric Brownian motion
\[
dZ_{t}=\left(\kappa_{t}+\lambda_{t}^{2}\right)Z_{t}dt-\lambda_{t}Z_{t}dB_{t}
\]
ie.
\[
Z_{t}=\exp\left(\int_{0}^{t}\left(\kappa_{s}+\frac{1}{2}\lambda_{s}^{2}\right)ds-\int_{0}^{t}\lambda_{s}dB_{s}\right)\,,
\]
to derive the following unconditionally stable discretization scheme
(called ``Pathwise Adapted Linearization'' in \citet{Kahl06})
\begin{eqnarray*}
\delta_{n} & \leftarrow & \left(\kappa_{t_{n}}+\frac{1}{2}\lambda_{t_{n}}^{2}\right)\Delta t_{n}-\lambda_{t_{n}}\Delta B_{n}\\
V_{t_{n+1}} & \leftarrow & V_{t_{n}}e^{-\delta_{n}}+\kappa_{t_{n}}\theta_{t_{n}}\frac{1-e^{-\delta_{n}}}{\delta_{n}}\Delta t_{n}
\end{eqnarray*}
where $0=t_{0}\leq\ldots\leq t_{n}\leq\ldots\leq t_{N}=T$ is a time
discretization of the interval $\left[0,T\right]$, with $\Delta t_{n}:=t_{n+1}-t_{n}$
and $\Delta B_{n}:=B_{t_{n+1}}-B_{t_{n}}$. This scheme ensures in
particular that the paths of $V$ remain positive.

\section{Numerical experiments\label{sec:Numerical_experiments}}

This section provides numerical tests of the fast pricing method proposed
in this paper. The Inverse Gamma stochastic volatility model will
be calibrated to market data using the pricing formula \eqref{eq:IGa_put_expansion}.
Three foreign exchange data sets will be used, detailed in subsection
\ref{sub:Datasets}. For each test case, we will provide the calibration
errors and expansion errors for implied volatility surfaces (Subsection
\ref{sub:Calibration}).

\subsection{Datasets\label{sub:Datasets}}

Three full sets of foreign exchange market data are provided here
for easy benchmarking. We provide the strikes used for the implied
volatility surface, and the equivalent constant rates for each maturity
(the constant rate $r_{eq}(T)$ equivalent to a time-dependent rate
$r(t)$, $0\leq t\leq T$ is defined by $r_{eq}(T):=\frac{1}{T}\int_{0}^{T}r(t)dt$).
The corresponding market implied volatility surfaces will be provided
in subsection \ref{sub:Calibration} (Tables \ref{tab:AUDUSD_implied_volatility},
\ref{tab:USDJPY_implied_volatility} and \ref{tab:USDSGD_implied_volatility}),
along with the corresponding calibration errors and expansion errors.
All the numbers quoted here have been rounded.

\subsubsection{Set 1: AUD/USD 17 June 2014 ($S_{0}=0.9335$)\vspace{-1.8mm}
}

\begin{table}[H]
\noindent \begin{centering}
\begin{tabular}{llllllllll}
\multicolumn{6}{l}{Strikes} & \hspace{10mm} & \multicolumn{3}{l}{Rates}\tabularnewline
\cline{1-6} \cline{8-10} 
Mat & 10 Put & 25 Put & ATM & 25 Call & 10 Call &  & Mat & foreign & domestic\tabularnewline
\cline{1-6} \cline{8-10} 
\noalign{\vskip1mm}
1M & 0.9103 & 0.9233 & 0.9356 & 0.9469 & 0.9572 &  & 1M & 2.80\% & 0.21\%\tabularnewline
\noalign{\vskip1mm}
3M & 0.8906 & 0.9168 & 0.9401 & 0.9605 & 0.9795 &  & 3M & 2.89\% & 0.31\%\tabularnewline
\noalign{\vskip1mm}
6M & 0.8664 & 0.9100 & 0.9469 & 0.9780 & 1.0078 &  & 6M & 3.03\% & 0.45\%\tabularnewline
\noalign{\vskip1mm}
1Y & 0.8322 & 0.9027 & 0.9609 & 1.0096 & 1.0580 &  & 1Y & 3.26\% & 0.69\%\tabularnewline[1mm]
\cline{1-6} \cline{8-10} 
\end{tabular}
\par\end{centering}

\protect\caption{Market data, AUDUSD, 17 June 2014}
\end{table}

\subsubsection{Set 2: USD/JPY 11 June 2014 ($S_{0}=102.00$)\vspace{-1.8mm}
}

\begin{table}[H]
\noindent \begin{centering}
\begin{tabular}{llllllllll}
\multicolumn{6}{l}{Strikes} & \hspace{10mm} & \multicolumn{3}{l}{Rates}\tabularnewline
\cline{1-6} \cline{8-10} 
Mat & 10 Put & 25 Put & ATM & 25 Call & 10 Call &  & Mat & foreign & domestic\tabularnewline
\cline{1-6} \cline{8-10} 
\noalign{\vskip1mm}
1M & 99.78 & 100.88 & 101.99 & 103.09 & 104.16 &  & 1M & 0.20\% & -0.04\%\tabularnewline
\noalign{\vskip1mm}
3M & 97.47 & \enskip{}99.77 & 101.98 & 104.15 & 106.31 &  & 3M & 0.29\% & \hspace{0.333em}0.07\%\tabularnewline
\noalign{\vskip1mm}
6M & 94.75 & \enskip{}98.47 & 102.00 & 105.46 & 109.06 &  & 6M & 0.40\% & \hspace{0.333em}0.16\%\tabularnewline
\noalign{\vskip1mm}
1Y & 90.04 & \enskip{}96.34 & 102.01 & 107.67 & 114.06 &  & 1Y & 0.52\% & \hspace{0.333em}0.21\%\tabularnewline[1mm]
\cline{1-6} \cline{8-10} 
\end{tabular}
\par\end{centering}

\protect\caption{Market data, USDJPY, 11 June 2014}
\end{table}

\subsubsection{Set 3: USD/SGD 04 September 2014 ($S_{0}=1.2541$)\vspace{-1.8mm}
}

\begin{table}[H]
\noindent \begin{centering}
\begin{tabular}{llllllllll}
\multicolumn{6}{l}{Strikes} & \hspace{10mm} & \multicolumn{3}{l}{Rates}\tabularnewline
\cline{1-6} \cline{8-10} 
Mat & 10 Put & 25 Put & ATM & 25 Call & 10 Call &  & Mat & foreign & domestic\tabularnewline
\cline{1-6} \cline{8-10} 
\noalign{\vskip1mm}
1M & 1.2397 & 1.2466 & 1.2542 & 1.2637 & 1.2755 &  & 1M & 0.16\% & 0.17\%\tabularnewline
\noalign{\vskip1mm}
2M & 1.2334 & 1.2432 & 1.2542 & 1.2688 & 1.2871 &  & 2M & 0.19\% & 0.19\%\tabularnewline
\noalign{\vskip1mm}
3M & 1.2286 & 1.2406 & 1.2543 & 1.2729 & 1.2970 &  & 3M & 0.28\% & 0.27\%\tabularnewline
\noalign{\vskip1mm}
6M & 1.2152 & 1.2339 & 1.2545 & 1.2836 & 1.3233 &  & 6M & 0.48\% & 0.47\%\tabularnewline
\noalign{\vskip1mm}
1Y & 1.1945 & 1.2232 & 1.2548 & 1.3018 & 1.3704 &  & 1Y & 0.61\% & 0.57\%\tabularnewline[1mm]
\cline{1-6} \cline{8-10} 
\end{tabular}
\par\end{centering}

\protect\caption{Market data, USDSGD, 04 September 2014}
\end{table}

Our empirical analyses suggest that these three data sets are representative
of the behavior on market data of the model and of the expansion.
In particular, our empirical findings will be similar for each data
set, leading us to present only three examples, as providing more
examples would not provide much additional information.

\subsection{Calibration\label{sub:Calibration}}

Using the closed-form expansion \eqref{eq:IGa_put_expansion}, we
can respectively calibrate the Inverse Gamma model to the three market
data sets. The four (piecewise-constant) stochastic parameters $\kappa$,
$\theta$, $\lambda$ and $\rho$ are calibrated, along with the initial
volatility $V_{0}$. This calibration process can be assessed from
the implied volatility calibration error. Then, using these parameters,
we will estimate the expansion error, by comparing the implied volatility
provided by the expansion \eqref{eq:IGa_put_expansion} to a Monte
Carlo price (Subsection \ref{sub:Monte_Carlo}), computed with $24$
time steps per day and $M=1\,000\,000$ paths (to keep both bias and
variance very low).

\subsubsection{Set 1: AUD/USD 17 June 2014 ($S_{0}=0.9335$)}

The estimated initial volatility is $V_{0}=6.49\%$, the estimated
piecewise-constant stochastic parameters are given by

\noindent \begin{center}
\begin{tabular}{lrrrr}
\noalign{\vskip1mm}
 & \multicolumn{1}{l}{$\kappa$} & \multicolumn{1}{l}{\enskip{}$\theta$} & \multicolumn{1}{l}{$\lambda$} & \multicolumn{1}{l}{\ $\rho$}\tabularnewline
\hline 
\noalign{\vskip1mm}
1M & 4.19 & 6.39\% & 1.71 & -0.40\tabularnewline
\noalign{\vskip1mm}
3M & 2.33 & 11.01\% & 1.12 & -0.74\tabularnewline
\noalign{\vskip1mm}
6M & 2.26 & 11.85\% & 1.25 & -0.73\tabularnewline
\noalign{\vskip1mm}
1Y & 1.80 & 12.52\% & 0.87 & -0.92\tabularnewline[1mm]
\hline 
\end{tabular},
\par\end{center}

and the calibration and expansion errors are given below (rounded
to the nearest basis point):

\begin{table}[H]
\begin{tabular}{lrrrrr}
 & \multicolumn{1}{l}{\enskip{}10 Put} & \multicolumn{1}{l}{25 Put} & \multicolumn{1}{l}{ATM} & \multicolumn{1}{l}{25 Call} & \multicolumn{1}{l}{10 Call}\tabularnewline
\hline 
\noalign{\vskip1mm}
1M & 7.48 {\small \textbf{{\color{gr1}[-0.04]}{\color{gr1}[-0.04]}}} & 6.87 {\small \textbf{{\color{gr1}[-0.02]}{\color{gr1}[-0.01]}}} & 6.38 {\small \textbf{{\color{gr2}[ 0.06]}{\color{gr1}[ 0.00]}}} & 6.19 {\small \textbf{{\color{gr1}[ 0.00]}{\color{gr1}[ 0.02]}}} & 6.19 {\small \textbf{{\color{gr1}[-0.03]}{\color{gr1}[-0.00]}}}\tabularnewline
\noalign{\vskip1mm}
3M & 8.46 {\small \textbf{{\color{gr2}[ 0.05]}{\color{gr3}[-0.12]}}} & 7.48 {\small \textbf{{\color{gr2}[-0.05]}{\color{gr1}[-0.03]}}} & 6.68 {\small \textbf{{\color{gr1}[ 0.03]}{\color{gr1}[ 0.03]}}} & 6.36 {\small \textbf{{\color{gr1}[-0.04]}{\color{gr2}[ 0.08]}}} & 6.39 {\small \textbf{{\color{gr2}[ 0.05]}{\color{gr2}[-0.08]}}}\tabularnewline
\noalign{\vskip1mm}
6M & 9.93 {\small \textbf{{\color{gr1}[ 0.02]}{\color{gr4}[-0.15]}}} & 8.43 {\small \textbf{{\color{gr1}[-0.03]}{\color{gr1}[-0.03]}}} & 7.30 {\small \textbf{{\color{gr2}[ 0.06]}{\color{gr1}[ 0.04]}}} & 6.82 {\small \textbf{{\color{gr2}[-0.07]}{\color{gr3}[ 0.13]}}} & 6.90 {\small \textbf{{\color{gr2}[ 0.08]}{\color{gr5}[-0.29]}}}\tabularnewline
\noalign{\vskip1mm}
1Y & 11.51 {\small \textbf{{\color{gr2}[-0.05]}{\color{gr4}[-0.19]}}} & 9.53 {\small \textbf{{\color{gr1}[ 0.00]}{\color{gr1}[-0.04]}}} & 8.05 {\small \textbf{{\color{gr4}[ 0.16]}{\color{gr2}[ 0.07]}}} & 7.47 {\small \textbf{{\color{gr3}[-0.14]}{\color{gr4}[ 0.19]}}} & 7.57 {\small \textbf{{\color{gr4}[ 0.15]}{\color{gr6}[-0.63]}}}\tabularnewline[1mm]
\hline 
\end{tabular}

\protect\caption{AUDUSD Market implied volatility {[}calibration error{]} {[}expansion
error{]} in \%\label{tab:AUDUSD_implied_volatility}}
\end{table}

The median absolute deviation of the calibration error is $5.0$bp,
and its mean absolute deviation is $5.7$bp.

The median absolute deviation of the expansion error is $5.5$bp,
and its mean absolute deviation is $10.9$bp.

Overall, the median absolute deviation of the total error is $6.0$bp,
and its mean absolute deviation is $10.5$bp.

\subsubsection{Set 2: USD/JPY 11 June 2014 ($S_{0}=102.00$)}

The estimated initial volatility is $V_{0}=4.42\%$, the estimated
piecewise-constant stochastic parameters are given by

\noindent \begin{center}
\begin{tabular}{lrrrr}
\noalign{\vskip1mm}
 & \multicolumn{1}{l}{$\kappa$} & \multicolumn{1}{l}{\enskip{}$\theta$} & \multicolumn{1}{l}{$\lambda$} & \multicolumn{1}{l}{\ $\rho$}\tabularnewline
\hline 
\noalign{\vskip1mm}
1M & 8.23 & 7.96\% & 2.47 & -0.10\tabularnewline
\noalign{\vskip1mm}
3M & 5.00 & 6.47\% & 1.32 & -0.19\tabularnewline
\noalign{\vskip1mm}
6M & 3.62 & 9.32\% & 1.61 & -0.15\tabularnewline
\noalign{\vskip1mm}
1Y & 2.10 & 6.74\% & 1.88 & -0.22\tabularnewline[1mm]
\hline 
\end{tabular},
\par\end{center}

and the calibration and expansion errors are given below (rounded
to the nearest basis point):

\begin{table}[H]
\begin{tabular}{lrrrrr}
 & \multicolumn{1}{l}{10 Put} & \multicolumn{1}{l}{25 Put} & \multicolumn{1}{l}{ATM} & \multicolumn{1}{l}{25 Call} & \multicolumn{1}{l}{10 Call}\tabularnewline
\hline 
\noalign{\vskip1mm}
1M & 6.08 {\small \textbf{{\color{gr2}[ 0.08]}{\color{gr1}[-0.00]}}} & 5.76 {\small \textbf{{\color{gr1}[ 0.00]}{\color{gr1}[-0.01]}}} & 5.53 {\small \textbf{{\color{gr1}[ 0.01]}{\color{gr1}[ 0.00]}}} & 5.63 {\small \textbf{{\color{gr2}[-0.07]}{\color{gr1}[ 0.00]}}} & 5.81 {\small \textbf{{\color{gr1}[-0.02]}{\color{gr1}[-0.00]}}}\tabularnewline
\noalign{\vskip1mm}
3M & 7.08 {\small \textbf{{\color{gr1}[ 0.02]}{\color{gr1}[-0.02]}}} & 6.53 {\small \textbf{{\color{gr1}[-0.03]}{\color{gr1}[-0.03]}}} & 6.15 {\small \textbf{{\color{gr1}[ 0.01]}{\color{gr1}[ 0.00]}}} & 6.20 {\small \textbf{{\color{gr1}[-0.03]}{\color{gr1}[ 0.01]}}} & 6.44 {\small \textbf{{\color{gr1}[ 0.04]}{\color{gr1}[-0.01]}}}\tabularnewline
\noalign{\vskip1mm}
6M & 8.19 {\small \textbf{{\color{gr3}[-0.12]}{\color{gr2}[ 0.06]}}} & 7.42 {\small \textbf{{\color{gr1}[ 0.02]}{\color{gr1}[-0.03]}}} & 6.95 {\small \textbf{{\color{gr3}[ 0.10]}{\color{gr1}[-0.01]}}} & 7.00 {\small \textbf{{\color{gr1}[ 0.02]}{\color{gr1}[ 0.00]}}} & 7.36 {\small \textbf{{\color{gr2}[-0.05]}{\color{gr1}[ 0.02]}}}\tabularnewline
\noalign{\vskip1mm}
1Y & 9.90 {\small \textbf{{\color{gr4}[-0.15]}{\color{gr5}[ 0.32]}}} & 8.61 {\small \textbf{{\color{gr3}[ 0.13]}{\color{gr3}[-0.12]}}} & 7.95 {\small \textbf{{\color{gr2}[ 0.07]}{\color{gr2}[-0.07]}}} & 8.04 {\small \textbf{{\color{gr1}[-0.04]}{\color{gr1}[-0.02]}}} & 8.69 {\small \textbf{{\color{gr2}[-0.06]}{\color{gr2}[ 0.08]}}}\tabularnewline[1mm]
\hline 
\end{tabular}

\protect\caption{USDJPY Market implied volatility {[}calibration error{]} {[}expansion
error{]} in \%\label{tab:USDJPY_implied_volatility}}
\end{table}

The median absolute deviation of the calibration error is $4.0$bp,
and its mean absolute deviation is $5.4$bp.

The median absolute deviation of the expansion error is $1.5$bp,
and its mean absolute deviation is $4.1$bp.

Overall, the median absolute deviation of the total error is $2.0$bp,
and its mean absolute deviation is $3.9$bp.

\subsubsection{Set 3: USD/SGD 04 September 2014 ($S_{0}=1.2541$)}

The estimated initial volatility is $V_{0}=3.16\%$, the estimated
piecewise-constant stochastic parameters are given by

\noindent \begin{center}
\begin{tabular}{lrrrr}
\noalign{\vskip1mm}
 & \multicolumn{1}{l}{$\kappa$} & \multicolumn{1}{l}{\enskip{}$\theta$} & \multicolumn{1}{l}{$\lambda$} & \multicolumn{1}{l}{\ $\rho$}\tabularnewline
\hline 
\noalign{\vskip1mm}
1M & 2.90 & 4.03\% & 2.30 & 0.49\tabularnewline
\noalign{\vskip1mm}
2M & 2.88 & 4.19\% & 1.64 & 0.49\tabularnewline
\noalign{\vskip1mm}
3M & 2.85 & 4.44\% & 2.37 & 0.58\tabularnewline
\noalign{\vskip1mm}
6M & 2.76 & 4.08\% & 1.68 & 0.51\tabularnewline
\noalign{\vskip1mm}
1Y & 2.81 & 4.27\% & 2.31 & 0.67\tabularnewline[1mm]
\hline 
\end{tabular},
\par\end{center}

and the calibration and expansion errors are given below (rounded
to the nearest basis point):

\begin{table}[H]
\begin{tabular}{lrrrrr}
 & \multicolumn{1}{l}{10 Put} & \multicolumn{1}{l}{25 Put} & \multicolumn{1}{l}{ATM} & \multicolumn{1}{l}{25 Call} & \multicolumn{1}{l}{10 Call}\tabularnewline
\hline 
\noalign{\vskip1mm}
1M & 3.06 {\small \textbf{{\color{gr2}[-0.07]}{\color{gr2}[ 0.05]}}} & 3.02 {\small \textbf{{\color{gr2}[ 0.06]}{\color{gr1}[ 0.02]}}} & 3.30 {\small \textbf{{\color{gr1}[ 0.04]}{\color{gr1}[ 0.01]}}} & 3.79 {\small \textbf{{\color{gr1}[ 0.01]}{\color{gr1}[-0.03]}}} & 4.43 {\small \textbf{{\color{gr2}[-0.05]}{\color{gr1}[-0.02]}}}\tabularnewline
\noalign{\vskip1mm}
2M & 3.15 {\small \textbf{{\color{gr1}[-0.02]}{\color{gr2}[ 0.07]}}} & 3.15 {\small \textbf{{\color{gr1}[ 0.02]}{\color{gr2}[ 0.05]}}} & 3.45 {\small \textbf{{\color{gr1}[ 0.01]}{\color{gr1}[-0.00]}}} & 4.10 {\small \textbf{{\color{gr1}[-0.02]}{\color{gr2}[-0.07]}}} & 4.84 {\small \textbf{{\color{gr1}[ 0.00]}{\color{gr1}[-0.03]}}}\tabularnewline
\noalign{\vskip1mm}
3M & 3.23 {\small \textbf{{\color{gr1}[-0.00]}{\color{gr2}[ 0.07]}}} & 3.24 {\small \textbf{{\color{gr1}[ 0.00]}{\color{gr2}[ 0.06]}}} & 3.57 {\small \textbf{{\color{gr1}[ 0.02]}{\color{gr1}[-0.00]}}} & 4.35 {\small \textbf{{\color{gr1}[-0.02]}{\color{gr2}[-0.09]}}} & 5.20 {\small \textbf{{\color{gr1}[-0.01]}{\color{gr1}[ 0.02]}}}\tabularnewline
\noalign{\vskip1mm}
6M & 3.52 {\small \textbf{{\color{gr1}[-0.03]}{\color{gr2}[ 0.08]}}} & 3.47 {\small \textbf{{\color{gr1}[-0.02]}{\color{gr2}[ 0.09]}}} & 3.80 {\small \textbf{{\color{gr2}[ 0.05]}{\color{gr1}[ 0.02]}}} & 4.78 {\small \textbf{{\color{gr1}[ 0.02]}{\color{gr3}[-0.11]}}} & 5.85 {\small \textbf{{\color{gr2}[-0.07]}{\color{gr4}[ 0.17]}}}\tabularnewline
\noalign{\vskip1mm}
1Y & 3.84 {\small \textbf{{\color{gr2}[-0.05]}{\color{gr1}[ 0.04]}}} & 3.78 {\small \textbf{{\color{gr1}[-0.04]}{\color{gr3}[ 0.11]}}} & 4.20 {\small \textbf{{\color{gr3}[ 0.12]}{\color{gr1}[ 0.03]}}} & 5.40 {\small \textbf{{\color{gr4}[ 0.15]}{\color{gr2}[-0.07]}}} & 6.78 {\small \textbf{{\color{gr4}[-0.19]}{\color{gr6}[ 0.68]}}}\tabularnewline[1mm]
\hline 
\end{tabular}

\protect\caption{USDSGD Market implied volatility {[}calibration error{]} {[}expansion
error{]} in \%\label{tab:USDSGD_implied_volatility}}
\end{table}

The median absolute deviation of the calibration error is $2.0$bp,
and its mean absolute deviation is $4.4$bp.

The median absolute deviation of the expansion error is $5.0$bp,
and its mean absolute deviation is $8.0$bp.

Overall, the median absolute deviation of the total error is $7.0$bp,
and its mean absolute deviation is $7.8$bp.

\subsection{Comments}

As shown by Tables \ref{tab:AUDUSD_implied_volatility}, \ref{tab:USDJPY_implied_volatility}
and \ref{tab:USDSGD_implied_volatility}, both calibration errors
and expansion errors are small overall (only a few basis points of
absolute deviation).

The overall calibrations are very good. The worst cases occur for
1Y maturity and 25 Call strike on Set 1 ($+19$bp), for 1Y maturity
and 10 Put strike on Set 2 ($-15$bp), and for 1Y maturity and 10
Call strike on Set 3 ($-19$bp). In other words, they occur for long
maturities and strikes far out of the money.

As to the expansion error, it is as expected very small for short
maturities and close to the money (ATM), but can get larger for long
maturities and far from the money. The worst cases occur for 1Y maturity
and 10 Call strike on Set 1 ($-63$bp), for 1Y maturity and 10 Put
strike on Set 2 ($+32$bp), and again for 1Y maturity and 10 Call
strike on Set 3 ($+68$bp). Though these worst case errors can look
large in implied volatility terms, the absolute option price values
are very small far out of the money, meaning that the actual error
is not that large in absolute price terms. Overall, we observe that
the expansion error increases with maturity $T$ and volatility of
volatility $\lambda$ (as expected from \citet{Benhamou10}), and
with absolute correlation $\left|\rho\right|$ (compare Set 2 to the
two other sets).

To decrease the calibration errors, one could try more general stochastic
volatility models, as discussed in Subsection \ref{sub:litterature},
though adding too many parameters may generate overfitting and damage
the stability and robustness of the model.

To decrease the expansion error, the most obvious solution is to compute
a higher-order expansion (expansion \eqref{eq:IGa_put_expansion}
is only a second-order expansion). Another idea is to take advantage
of the fact that the expansion error is much smaller in practice when
the stochastic parameters are constant (only a few bps everywhere).
\citet{Benhamou10} tried to exploit this idea by first calibrating
a model with constant parameters for each maturity, and then turning
these calibrations into an ``equivalent'' model with piecewise constant
parameters. However, the equivalent piecewise constant model is determined
using the same second-order expansion, therefore this two-step calibration
procedure does not improve the expansion error of the final model
with piecewise-constant parameters.

Depending on the intended application, the accuracy of the present
second-order expansion \eqref{eq:IGa_put_expansion}, as observed
in Subsection \ref{sub:Calibration}, may be sufficient. For example,
if one only needs the calibrated stochastic parameters to feed into
a more general local-stochastic Inverse Gamma volatility model, then
the local volatility component can easily eliminate the residual price
discrepancies (cf. \citet{Sepp14}).

\subsection{Comparison to Heston\label{sub:Heston_comparison}}

Finally, we compare the Inverse Gamma calibration to a classical Heston
stochastic volatility calibration. Subsection \ref{sub:Other-models}
provided theoretical reasons to favor the Inverse Gamma model over
the Heston model, so it is interesting to check if, for example, one
of the two models provides significantly better calibrations in practice,
as this is one of the important criteria for practical use by the
industry. We use the same three datasets, and calibrate the Heston
model using the semi-closed-form of \citet{Heston93}.

\subsubsection{Set 1: AUD/USD 17 June 2014 ($S_{0}=0.9335$)}

For the Heston model, the estimated initial variance is $V_{0}=0.41\%$,
the estimated piecewise constant stochastic parameters are given by

\noindent \begin{center}
\begin{tabular}{cccccc}
 & $\kappa$ & \enskip{}$\theta$ & $\lambda$ & \ $\rho$ & $2\kappa\theta/\lambda^{2}$\tabularnewline
\hline 
1M & 1.16 & 1.28\% & 0.32 & -0.32 & 0.30\tabularnewline
3M & 0.97 & 2.32\% & 0.48 & -0.49 & 0.20\tabularnewline
6M & 1.01 & 1.88\% & 0.51 & -0.54 & 0.15\tabularnewline
1Y & 1.02 & 1.85\% & 0.51 & -0.52 & 0.14\tabularnewline
\hline 
\end{tabular},
\par\end{center}

and the calibration error is given below (rounded to the nearest basis
point):

\noindent \begin{center}
\begin{table}[h]
\noindent \begin{centering}
\begin{tabular}{lrrrrr}
 & \multicolumn{1}{l}{\enskip{}10 Put} & \multicolumn{1}{l}{25 Put} & \multicolumn{1}{l}{ATM} & \multicolumn{1}{l}{25 Call} & \multicolumn{1}{l}{10 Call}\tabularnewline
\hline 
\noalign{\vskip1mm}
1M & 7.48 {\small \textbf{{\color{gr1}[ 0.02]}}} & 6.87 {\small \textbf{{\color{gr1}[-0.03]}}} & 6.38 {\small \textbf{{\color{gr1}[ 0.00]}}} & 6.19 {\small \textbf{{\color{gr1}[-0.04]}}} & 6.19 {\small \textbf{{\color{gr2}[-0.05]}}}\tabularnewline
\noalign{\vskip1mm}
3M & 8.46 {\small \textbf{{\color{gr1}[ 0.03]}}} & 7.48 {\small \textbf{{\color{gr2}[-0.07]}}} & 6.68 {\small \textbf{{\color{gr1}[ 0.01]}}} & 6.36 {\small \textbf{{\color{gr1}[ 0.03]}}} & 6.39 {\small \textbf{{\color{gr2}[ 0.07]}}}\tabularnewline
\noalign{\vskip1mm}
6M & 9.93 {\small \textbf{{\color{gr1}[ 0.03]}}} & 8.43 {\small \textbf{{\color{gr2}[-0.05]}}} & 7.30 {\small \textbf{{\color{gr1}[ 0.03]}}} & 6.82 {\small \textbf{{\color{gr1}[ 0.02]}}} & 6.90 {\small \textbf{{\color{gr2}[-0.06]}}}\tabularnewline
\noalign{\vskip1mm}
1Y & 11.51 {\small \textbf{{\color{gr1}[-0.04]}}} & 9.53 {\small \textbf{{\color{gr3}[-0.10]}}} & 8.05 {\small \textbf{{\color{gr3}[ 0.10]}}} & 7.47 {\small \textbf{{\color{gr2}[ 0.08]}}} & 7.57 {\small \textbf{{\color{gr2}[-0.07]}}}\tabularnewline[1mm]
\hline 
\end{tabular}
\par\end{centering}

\protect\caption{AUDUSD Market implied volatility {[}Heston calibration error{]} in
\%, \label{tab:AUDUSD_implied_volatility_H}}
\end{table}

\par\end{center}

The median absolute deviation of the calibration error is $4.3$bp,
and its mean absolute deviation is $4.6$bp.

\subsubsection{Set 2: USD/JPY 11 June 2014 ($S_{0}=102.00$)}

The estimated initial variance is $V_{0}=0.28\%$, the estimated piecewise
constant stochastic parameters are given by

\noindent \begin{center}
\begin{tabular}{cccccc}
 & $\kappa$ & \enskip{}$\theta$ & $\lambda$ & \ $\rho$ & $2\kappa\theta/\lambda^{2}$\tabularnewline
\hline 
1M & 1.17 & 1.39\% & 0.23 & -0.11 & 0.62\tabularnewline
3M & 1.10 & 1.76\% & 0.40 & -0.21 & 0.25\tabularnewline
6M & 1.09 & 1.73\% & 0.44 & -0.22 & 0.20\tabularnewline
1Y & 1.04 & 1.92\% & 0.48 & -0.43 & 0.17\tabularnewline
\hline 
\end{tabular},
\par\end{center}

and the calibration error is given below (rounded to the nearest basis
point):

\noindent \begin{center}
\begin{table}[h]
\noindent \begin{centering}
\begin{tabular}{lrrrrr}
 & \multicolumn{1}{l}{10 Put} & \multicolumn{1}{l}{25 Put} & \multicolumn{1}{l}{ATM} & \multicolumn{1}{l}{25 Call} & \multicolumn{1}{l}{10 Call}\tabularnewline
\hline 
\noalign{\vskip1mm}
1M & 6.08 {\small \textbf{{\color{gr2}[-0.06]}}} & 5.76 {\small \textbf{{\color{gr2}[-0.06]}}} & 5.53 {\small \textbf{{\color{gr1}[ 0.00]}}} & 5.63 {\small \textbf{{\color{gr3}[-0.10]}}} & 5.81 {\small \textbf{{\color{gr3}[-0.13]}}}\tabularnewline
\noalign{\vskip1mm}
3M & 7.08 {\small \textbf{{\color{gr1}[-0.00]}}} & 6.53 {\small \textbf{{\color{gr2}[-0.06]}}} & 6.15 {\small \textbf{{\color{gr1}[ 0.00]}}} & 6.20 {\small \textbf{{\color{gr1}[-0.02]}}} & 6.44 {\small \textbf{{\color{gr1}[ 0.04]}}}\tabularnewline
\noalign{\vskip1mm}
6M & 8.19 {\small \textbf{{\color{gr3}[ 0.13]}}} & 7.42 {\small \textbf{{\color{gr1}[-0.04]}}} & 6.95 {\small \textbf{{\color{gr2}[-0.05]}}} & 7.00 {\small \textbf{{\color{gr2}[-0.05]}}} & 7.36 {\small \textbf{{\color{gr2}[ 0.06]}}}\tabularnewline
\noalign{\vskip1mm}
1Y & 9.90 {\small \textbf{{\color{gr1}[ 0.04]}}} & 8.61 {\small \textbf{{\color{gr1}[-0.03]}}} & 7.95 {\small \textbf{{\color{gr1}[-0.02]}}} & 8.04 {\small \textbf{{\color{gr1}[-0.01]}}} & 8.69 {\small \textbf{{\color{gr1}[ 0.02]}}}\tabularnewline[1mm]
\hline 
\end{tabular}
\par\end{centering}

\protect\caption{USDJPY Market implied volatility {[}Heston calibration error{]} in
\%\label{tab:USDJPY_implied_volatility_H}}
\end{table}

\par\end{center}

The median absolute deviation of the calibration error is $4.0$bp,
and its mean absolute deviation is $4.6$bp.

\subsubsection{Set 3: USD/SGD 04 September 2014 ($S_{0}=1.2541$)}

The estimated initial variance is $V_{0}=0.11\%$, the estimated piecewise
constant stochastic parameters are given by

\noindent \begin{center}
\begin{tabular}{cccccc}
 & $\kappa$ & \enskip{}$\theta$ & $\lambda$ & \ $\rho$ & $2\kappa\theta/\lambda^{2}$\tabularnewline
\hline 
1M & 1.25 & 0.49\% & 0.28 & 0.41 & 0.16\tabularnewline
2M & 1.15 & 0.62\% & 0.49 & 0.29 & 0.06\tabularnewline
3M & 1.24 & 0.54\% & 0.48 & 0.23 & 0.06\tabularnewline
6M & 1.51 & 0.32\% & 0.52 & 0.44 & 0.04\tabularnewline
1Y & 1.21 & 0.86\% & 0.49 & -0.26 & 0.09\tabularnewline
\hline 
\end{tabular},
\par\end{center}

and the calibration error is given below (rounded to the nearest basis
point):

\noindent \begin{center}
\begin{table}[H]
\noindent \begin{centering}
\begin{tabular}{lrrrrr}
 & \multicolumn{1}{l}{10 Put} & \multicolumn{1}{l}{25 Put} & \multicolumn{1}{l}{ATM} & \multicolumn{1}{l}{25 Call} & \multicolumn{1}{l}{10 Call}\tabularnewline
\hline 
\noalign{\vskip1mm}
1M & 3.06 {\small \textbf{{\color{gr2}[-0.07]}}} & 3.02 {\small \textbf{{\color{gr1}[ 0.02]}}} & 3.30 {\small \textbf{{\color{gr1}[ 0.00]}}} & 3.79 {\small \textbf{{\color{gr1}[-0.02]}}} & 4.43 {\small \textbf{{\color{gr1}[-0.02]}}}\tabularnewline
\noalign{\vskip1mm}
2M & 3.15 {\small \textbf{{\color{gr2}[-0.05]}}} & 3.15 {\small \textbf{{\color{gr1}[-0.01]}}} & 3.45 {\small \textbf{{\color{gr1}[-0.01]}}} & 4.10 {\small \textbf{{\color{gr2}[-0.05]}}} & 4.84 {\small \textbf{{\color{gr1}[ 0.01]}}}\tabularnewline
\noalign{\vskip1mm}
3M & 3.23 {\small \textbf{{\color{gr1}[ 0.02]}}} & 3.24 {\small \textbf{{\color{gr1}[-0.00]}}} & 3.57 {\small \textbf{{\color{gr1}[-0.01]}}} & 4.35 {\small \textbf{{\color{gr3}[-0.10]}}} & 5.20 {\small \textbf{{\color{gr1}[-0.00]}}}\tabularnewline
\noalign{\vskip1mm}
6M & 3.52 {\small \textbf{{\color{gr1}[ 0.01]}}} & 3.47 {\small \textbf{{\color{gr1}[-0.03]}}} & 3.80 {\small \textbf{{\color{gr1}[ 0.03]}}} & 4.78 {\small \textbf{{\color{gr1}[-0.04]}}} & 5.85 {\small \textbf{{\color{gr3}[ 0.13]}}}\tabularnewline
\noalign{\vskip1mm}
1Y & 3.84 {\small \textbf{{\color{gr2}[ 0.09]}}} & 3.78 {\small \textbf{{\color{gr1}[-0.04]}}} & 4.20 {\small \textbf{{\color{gr1}[-0.02]}}} & 5.40 {\small \textbf{{\color{gr2}[-0.08]}}} & 6.78 {\small \textbf{{\color{gr3}[ 0.13]}}}\tabularnewline[1mm]
\hline 
\end{tabular}
\par\end{centering}

\protect\caption{USDSGD Market implied volatility {[}Heston calibration error{]} in
\%\label{tab:USDSGD_implied_volatility_H}}
\end{table}

\par\end{center}

The median absolute deviation of the calibration error is $2.5$bp,
and its mean absolute deviation is $4.0$bp.

\subsubsection{Result analysis and comparison}

For the three examples, the calibrated parameters do not satisfy the
Feller condition: the Feller ratio is much smaller than $1$. This
is the reason why we could not rely on the Heston closed-form expansion
of \citet{Benhamou10} for calibration. Indeed, when the Feller ratio
is so low, the expansion error can be massive. For the USD/JPY parameters,
the average absolute expansion error is around $80$bp. For the AUD/USD
and USD/SGD parameters, the error is so large that some expansion
prices can even become negative. When the implied volatility exists,
the average absolute expansion error is around $200$bp for both AUD/USD
and USD/SGD parameters. These large errors were expected, as the expansion
approach for the Heston model is not expected to work when the Feller
condition is not satisfied (the error analysis in \citet{Benhamou10}
requires the Feller condition to hold). This means that the expansion
scheme for the Heston model cannot be used in practice because, as
recalled in subsection \ref{sub:Other-models}, the Feller condition
is virtually never satisfied on real-world market data. For this reason,
we used the slower semi-closed-form pricing formula of Heston for
calibrating the model.

The table below summarizes the calibration error (median and mean,
in basis points) for both Inverse Gamma and Heston stochastic volatility
models.

\noindent \begin{center}
\begin{tabular}{lllllll}
 & \multicolumn{2}{l}{AUD/USD} & \multicolumn{2}{l}{USD/JPY} & \multicolumn{2}{l}{USD/SGD}\tabularnewline
\cline{2-7} 
calibration error (bp) & median & mean & median & mean & median & mean\tabularnewline
\hline 
Inverse Gamma & 5.0 & 5.7 & 4.0 & 5.4 & 2.0 & 4.4\tabularnewline
Heston & 4.3 & 4.6 & 4.0 & 4.6 & 2.5 & 4.0\tabularnewline
\hline 
\end{tabular}
\par\end{center}

One can see that the calibration error is of the same order. If anything,
the Heston calibration is slightly better ($-0.8$bp on average).
Therefore, the calibration error itself does not provide clear and
sufficient indication on which stochastic volatility model is a better
model in practice. To assess the quality of a stochastic volatility
model, the stability of the estimated parameters (the ``variance''
in statistical terms) is as relevant as the calibration error (the
``bias''' in statistical terms). Indeed, a very flexible stochastic
volatility model with many parameters is likely to have a very small
calibration error on a given market implied volatility curve. However,
the parameters of such a model are also likely to be very unstable
over time, if recalibrated every day as is common in practice. The
stability of parameters is an important criterion for dynamic hedging
to work in practice and for the option price to be meaningful at all.
This ``bias-variance tradeoff'' discussion suggests that more numerical
tests, in particular stability tests, are needed to further compare
the Inverse Gamma and the Heston stochastic volatility models. In
particular, a dynamic hedging backtest using the two models for several
options over a long time period would be useful. We plan to carry
out such tests in the future.

At this stage of the comparison, we recommend the Inverse Gamma model
over the Heston model for foreign exchange option pricing for at least
three reasons:
\begin{itemize}
\item Firstly, even though, based on the three examples studied in this
paper, the Heston calibration seems slightly more accurate, the fact
that the optimal parameters are very far from satisfying the Feller
condition is a major issue, as the very unrealistic volatility distribution
resulting from this breach (see Figure \ref{fig:vol_std_0.24}) is
likely to affect the quality of dynamic hedging and the pricing of
more exotic options.
\item Secondly, from our experience of using the model for more than forty
currency pairs, the closed-form expansion for Inverse Gamma vanilla
options (Theorem \ref{thm:IGa_expansion}) works quite well on real-world
market data (unlike the Heston expansion). If needed, the accuracy
of the expansion, already good, can always be enhanced using a higher-order
expansion (e.g. 4th order).
\item Lastly, the closed-form expansion approach is the fastest method for
calibration purposes, and therefore is highly desirable in practice
when calibration needs to be performed for hundreds of currency pairs
or securities on at least a daily basis. 
\end{itemize}

\section{Conclusion\label{sec:Conclusion}}

This paper has introduced the Inverse Gamma stochastic volatility
model, as defined by the volatility dynamics $dV_{t}=\kappa_{t}\left(\theta_{t}-V_{t}\right)dt+\lambda_{t}V_{t}dB_{t}$. 

The volatility distribution in this model is more consistent with
market dynamics than alternative one factor affine stochastic volatility
models such as the Heston model.

We have proposed a closed-form volatility of volatility expansion
for the price of a European put option under this stochastic volatility
model, and simple and straightforward recursion formulae to instantaneously
compute the coefficients of the expansion when the four time-dependent
stochastic parameters $\kappa$, $\theta$, $\lambda$ and $\rho$
are piecewise constant. 

We have demonstrated the viability of the second-order expansion scheme
on three test cases from foreign exchange (AUD/USD, USD/JPY and USD/SGD).
Both calibration error and expansion error are small overall (only
a few basis points of absolute deviation on average). 

Potential improvements can be made and have also been discussed in
the paper, such as introducing additional stochastic parameters, computing
higher order terms in the expansion, performing more comprehensive
calibration backtests, and studying the local-stochastic volatility
version of this model. We hope to foster further academic research
on these non-affine models favoured by industry practitioners, as
this paper demonstrates that they not only provide a more accurate
representation of market dynamics, but also still be tractable thanks
to expansion methods.

\section*{Acknowledgement}

The authors would like to thank Mr Julian Cook of GFI FENICS for many
insightful and valuable discussions, and for providing the market
data of this research work.

\bibliographystyle{abbrvnat}
\bibliography{Biblio}

\begin{thebibliography}{34}
\providecommand{\natexlab}[1]{#1}
\providecommand{\url}[1]{\texttt{#1}}
\expandafter\ifx\csname urlstyle\endcsname\relax
  \providecommand{\doi}[1]{doi: #1}\else
  \providecommand{\doi}{doi: \begingroup \urlstyle{rm}\Url}\fi

\bibitem[Andersen and Piterbarg(2007)]{Andersen07}
L.~Andersen and V.~Piterbarg.
\newblock Moment explosions in stochastic volatility models.
\newblock \emph{Finance and Stochastics}, 11\penalty0 (1):\penalty0 29--50,
  2007.

\bibitem[Barone-Adesi et~al.(2005)Barone-Adesi, Rasmussen, and
  Ravanelli]{Barone05}
G.~Barone-Adesi, H.~Rasmussen, and C.~Ravanelli.
\newblock An option pricing formula for the {GARCH} diffusion model.
\newblock \emph{Computational Statistics and Data Analysis}, 49\penalty0
  (2):\penalty0 287--310, 2005.

\bibitem[Benhamou et~al.(2010)Benhamou, Gobet, and Miri]{Benhamou10}
E.~Benhamou, E.~Gobet, and M.~Miri.
\newblock Time dependent {H}eston model.
\newblock \emph{SIAM Journal on Financial Mathematics}, 1\penalty0
  (1):\penalty0 289--325, 2010.

\bibitem[Bouchaud and Potters(2003)]{Bouchaud03}
J.-P. Bouchaud and M.~Potters.
\newblock \emph{Theory of Financial Risk and Derivative Pricing: From
  Statistical Physics to Risk Management}.
\newblock Cambridge University Press, 2003.

\bibitem[Christoffersen et~al.(2010)Christoffersen, Jacobs, and
  Mimouni]{Christoffersen10}
P.~Christoffersen, K.~Jacobs, and K.~Mimouni.
\newblock Volatility dynamics for the {S}\&{P}500: Evidence from realized
  volatility, daily returns, and option prices.
\newblock \emph{Review of Financial Studies}, 23\penalty0 (8):\penalty0
  3141--3189, 2010.

\bibitem[Clark(2011)]{Clark11}
I.~Clark.
\newblock \emph{Foreign exchange option pricing: {A} practitioners's guide}.
\newblock Wiley Finance. John Wiley \& Sons, 2011.

\bibitem[da~Fonseca and Grasselli(2011)]{Fonseca11}
J.~da~Fonseca and M.~Grasselli.
\newblock Riding on the smiles.
\newblock \emph{Quantitative Finance}, 11\penalty0 (11):\penalty0 1609--1632,
  2011.

\bibitem[Fornari and Mele(2001)]{Fornari01}
F.~Fornari and A.~Mele.
\newblock Recovering the probability density function of asset prices using
  {GARCH} as diffusion approximations.
\newblock \emph{Journal of Empirical Finance}, 8\penalty0 (1):\penalty0
  83--110, 2001.

\bibitem[Fornari and Mele(2006)]{Fornari06}
F.~Fornari and A.~Mele.
\newblock Approximating volatility diffusions with {CEV-ARCH} models.
\newblock \emph{Journal of Economic Dynamics and Control}, 30\penalty0
  (6):\penalty0 931--966, 2006.

\bibitem[Gander and Stephens(2007)]{Gander07}
M.~Gander and D.~Stephens.
\newblock Stochastic volatility modelling in continuous time with general
  marginal distributions: Inference, prediction and model selection.
\newblock \emph{Journal of Statistical Planning and Inference}, 137\penalty0
  (10):\penalty0 3068--3081, 2007.

\bibitem[Gatheral(2007)]{Gatheral07}
J.~Gatheral.
\newblock Further developments in volatility derivatives pricing.
\newblock In \emph{Global Derivatives}, 2007.

\bibitem[Gatheral(2008)]{Gatheral08}
J.~Gatheral.
\newblock Consistent modeling of {SPX} and {VIX} options.
\newblock In \emph{Bachelier Congress}, 2008.

\bibitem[Hansis(2010)]{Hansis10}
A.~Hansis.
\newblock Affine versus non-affine stochastic volatility and the impact on
  asset allocation.
\newblock SSRN:1545783, 2010.

\bibitem[Henry-Labord\`ere(2008)]{HenryLabordere08}
P.~Henry-Labord\`ere.
\newblock \emph{Analysis, Geometry, and Modeling in Finance: Advanced Methods
  in Option Pricing}.
\newblock Financial Mathematics Series. Chapman \& Hall/CRC, 2008.

\bibitem[Henry-Labord\`ere(2009)]{HenryLabordere09}
P.~Henry-Labord\`ere.
\newblock Calibration of local stochastic volatility models to market smiles.
\newblock \emph{Risk}, 22\penalty0 (9):\penalty0 112--117, 2009.

\bibitem[Heston(1993)]{Heston93}
S.~Heston.
\newblock A closed-form solution for options with stochastic volatility with
  applications to bond and currency options.
\newblock \emph{Review of Financial Studies}, 6\penalty0 (2):\penalty0
  327--343, 1993.

\bibitem[Itkin(2013)]{Itkin13}
A.~Itkin.
\newblock New solvable stochastic volatility models for pricing volatility
  derivatives.
\newblock \emph{Review of Derivatives Research}, 16\penalty0 (2):\penalty0
  111--134, 2013.

\bibitem[Jerbi(2011)]{Jerbi11}
Y.~Jerbi.
\newblock Methodology for stochastic volatility process calibration application
  to the {CAC} 40 index.
\newblock \emph{Journal of Statistical Computation and Simulation}, 83\penalty0
  (3):\penalty0 417--433, 2011.

\bibitem[Kaeck and Alexander(2012)]{Kaeck12}
A.~Kaeck and C.~Alexander.
\newblock Volatility dynamics for the {S}\&{P} 500: Further evidence from
  non-affine, multi-factor jump diffusions.
\newblock \emph{Journal of Banking \& Finance}, 36\penalty0 (11):\penalty0
  3110--3121, 2012.

\bibitem[Kahl and Jackel(2006)]{Kahl06}
C.~Kahl and P.~Jackel.
\newblock Fast strong approximation {M}onte {C}arlo schemes for stochastic
  volatility models.
\newblock \emph{Quantitative Finance}, 6\penalty0 (6):\penalty0 513--536, 2006.

\bibitem[Lewis(2000)]{Lewis00}
A.~Lewis.
\newblock \emph{Option valuation under stochastic volatility with Mathematica
  code}.
\newblock Finance Press, 2000.

\bibitem[Lo et~al.(2003)Lo, Lee, and Hui]{Lo03}
C.~Lo, H.~Lee, and C.~Hui.
\newblock A simple approach for pricing barrier options with time-dependent
  parameters.
\newblock \emph{Quantitative Finance}, 3\penalty0 (2):\penalty0 98--107, 2003.

\bibitem[Ma and Serota(2014)]{Ma14}
T.~Ma and R.~Serota.
\newblock A model for stock returns and volatility.
\newblock \emph{Physica A: Statistical Mechanics and its Applications},
  398:\penalty0 89--115, 2014.

\bibitem[Rapisarda(2003)]{Rapisarda03}
F.~Rapisarda.
\newblock Pricing barriers on underlyings with time-dependent parameters.
\newblock Technical report, Banca IMI, 2003.

\bibitem[Ribeiro and Poulsen(2013)]{Ribeiro13}
A.~Ribeiro and R.~Poulsen.
\newblock Approximation behoves calibration.
\newblock \emph{Quantitative Finance Letters}, 1\penalty0 (1):\penalty0 36--40,
  2013.

\bibitem[Rojo(1996)]{Rojo96}
J.~Rojo.
\newblock On tail categorization of probability laws.
\newblock \emph{Journal of the American Statistical Association}, 91\penalty0
  (433):\penalty0 378--384, 1996.

\bibitem[Sch\"obel and Zhu(1999)]{Schobel99}
R.~Sch\"obel and J.~Zhu.
\newblock Stochastic volatility with an {O}rnstein-{U}hlenbeck process: an
  extension.
\newblock \emph{European Finance Review}, 3\penalty0 (1):\penalty0 23--46,
  1999.

\bibitem[Sepp(2014)]{Sepp14}
A.~Sepp.
\newblock Empirical calibration and minimum-variance delta under log-normal
  stochastic volatility dynamics.
\newblock SSRN:2387845, 2014.

\bibitem[Sepp(2015)]{Sepp15}
A.~Sepp.
\newblock Log-normal stochastic volatility model: Pricing of vanilla options
  and econometric estimation.
\newblock SSRN:2522425, 2015.

\bibitem[Shiraya and Takahashi(2011)]{Shiraya11}
K.~Shiraya and A.~Takahashi.
\newblock Pricing average options on commodities.
\newblock \emph{Journal of Futures Markets}, 31\penalty0 (5):\penalty0
  407--439, 2011.

\bibitem[Shiraya and Takahashi(2014)]{Shiraya14}
K.~Shiraya and A.~Takahashi.
\newblock Pricing multiasset cross-currency options.
\newblock \emph{Journal of Futures Markets}, 34\penalty0 (1):\penalty0 1--19,
  2014.

\bibitem[Tataru and Fisher(2012)]{Tataru12}
G.~Tataru and T.~Fisher.
\newblock The {B}loomberg stochastic local volatility model for {FX} exotics.
\newblock Technical report, Bloomberg, 2012.

\bibitem[Wiggins(1987)]{Wiggins87}
J.~Wiggins.
\newblock Option values under stochastic volatility: Theory and empirical
  estimates.
\newblock \emph{Journal of Financial Economics}, 19\penalty0 (2):\penalty0
  351--372, 1987.

\bibitem[Zhao(2009)]{Zhao09}
B.~Zhao.
\newblock Inhomogeneous {G}eometric {B}rownian {M}otions.
\newblock SSRN:1429449, 2009.

\end{thebibliography}

\appendix

\section{Stationary distribution of volatility\label{sec:stationary_vol}}

\subsection{Heston}

For the Heston stochastic volatility model with constant coefficients,
\begin{eqnarray}
dS_{t} & = & (r_{d}-r_{f})S_{t}dt+\sqrt{V_{t}}S_{t}dW_{t}\label{eq:Heston}\\
dV_{t} & = & \kappa\left(\theta-V_{t}\right)dt+\lambda\sqrt{V_{t}}dB_{t}\nonumber \\
d\left\langle W,B\right\rangle _{t} & = & \rho dt\,,
\end{eqnarray}

the long-run distribution of the variance $V_{t}$ is a Gamma distribution
with parameters $k_{\Gamma}=\beta\theta$ and $\theta_{\Gamma}=\frac{1}{\beta}$,
where $\beta:=\frac{2\kappa}{\lambda^{2}}$. Remark that the probability
density function of $V_{t}$ can reach $0$ if $\beta\theta<1$ (Feller
condition). 

Consequently, the long-run Heston volatility $Y_{t}=\sqrt{V_{t}}$
has a generalized Chi distribution
\[
p_{\chi}\left(x;a,b,\nu\right)=\frac{1}{2^{\frac{\nu}{2}-1}b\Gamma\left(\frac{\nu}{2}\right)}\left(\frac{x-a}{b}\right)^{\nu-1}e^{-\frac{1}{2}\left(\frac{x-a}{b}\right)^{2}}
\]
with parameters $a=0$, $b=\frac{1}{\sqrt{2\beta}}$ and $\nu=2\beta\theta$. 

In particular, the moments of the long-run volatility $Y_{t}=\sqrt{V_{t}}$
are given by:
\begin{eqnarray}
\mathbb{E}\left[Y_{t}\right] & \rightarrow & \frac{\Gamma\left(k_{\Gamma}+\frac{1}{2}\right)}{\Gamma\left(k_{\Gamma}\right)}\sqrt{\theta_{\Gamma}}=\frac{\Gamma\left(\beta\theta+\frac{1}{2}\right)}{\Gamma\left(\beta\theta\right)\sqrt{\beta\theta}}\sqrt{\theta}\nonumber \\
\mathbb{E}\left[Y_{t}^{2}\right] & \rightarrow & k_{\Gamma}\theta_{\Gamma}=\theta\label{eq:moments_Heston_vol}
\end{eqnarray}

\subsection{Inverse Gamma \label{sub:stationary_IGa}}

For the Inverse Gamma stochastic volatility model with constant coefficients,
\begin{eqnarray*}
dS_{t} & = & (r_{d}-r_{f})S_{t}dt+V_{t}S_{t}dW_{t}\\
dV_{t} & = & \kappa\left(\theta-V_{t}\right)dt+\lambda V_{t}dB_{t}\\
d\left\langle W,B\right\rangle _{t} & = & \rho dt\,,
\end{eqnarray*}
the long-run distribution of the volatility process $V_{t}$ is an
inverse Gamma distribution
\[
p_{\Gamma^{-1}}\left(x;\alpha,\beta\right)=\frac{\beta^{\alpha}}{\Gamma\left(\alpha\right)}x^{-\alpha-1}e^{-\frac{\beta}{x}}
\]
with parameters $\alpha_{\Gamma^{-1}}=1+\beta$ and $\beta_{\Gamma^{-1}}=\beta\theta$,
where $\beta:=\frac{2\kappa}{\lambda^{2}}$ (see \citet{Barone05}).
Therefore
\begin{eqnarray}
\mathbb{E}\left[V_{t}\right] & \rightarrow & \frac{\beta_{\Gamma^{-1}}}{\alpha_{\Gamma^{-1}}-1}=\theta\,\,\,(iff\,\beta>0)\nonumber \\
\mathrm{Var}\left[V_{t}\right] & \rightarrow & \frac{\beta_{\Gamma^{-1}}^{2}}{\left(\alpha_{\Gamma^{-1}}-1\right)^{2}\left(\alpha_{\Gamma^{-1}}-2\right)}=\frac{\theta^{2}}{\beta-1}\,\,\,(iff\,\beta>1)\label{eq:moments_IGa_vol}
\end{eqnarray}

\section{The Log-Normal terminology\label{sec:log_normal_really}}

There exists three different volatility dynamics that have been called
log-normal in the literature:
\begin{eqnarray}
dV_{t} & = & \kappa V_{t}dt+\lambda V_{t}dB_{t}\label{eq:LN-GBM}\\
dV_{t} & = & \kappa\left(\theta-V_{t}\right)dt+\lambda V_{t}dB_{t}\label{eq:LN-IGa}\\
d\log(V_{t}) & = & \kappa\left(\theta-\log(V_{t})\right)dt+\lambda dB_{t}\label{eq:LN-ExpOU}
\end{eqnarray}
In the first formulation \eqref{eq:LN-GBM}, the volatility is modeled
by a geometric Brownian motion, which is log-normally distributed.
It is a special case of SABR model (with $\beta=1$). It is not mean-reverting.

In the last formulation \eqref{eq:LN-ExpOU}, the logarithm of the
volatility is an Ornstein-Uhlenbeck process, which has a normally-distributed
stationary distribution. Therefore the long-term distribution of the
volatility $V_{t}$ is indeed log-normally distributed.

The intermediate formulation \eqref{eq:LN-IGa} combines characteristics
from a geometric Brownian motion (the volatility of volatility $\lambda V_{t}$
is proportional to the volatility $V_{t}$) and from an exponential
Ornstein-Uhlenbeck (the mean-reversion effect towards a level $\theta$).
However the stationary distribution of the volatility is not a log-normal
distribution but an inverse Gamma distribution (\citet{Barone05,Zhao09,Sepp14,Sepp15}).

Practitioners sometimes refer to \eqref{eq:LN-IGa} as a ``Log-Normal''
or ``mean-reverting Log-Normal'' stochastic volatility model, but
as this terminology can be ambiguous and misleading, we choose to
call it Inverse Gamma stochastic volatility model, which is consistent
with the stationary distribution of the volatility \eqref{eq:LN-IGa}.

\section{Proof of main expansion\label{sec:proof}}

In the IGa model \eqref{eq:IGa}, factoring out the drift rates $r_{d}$
and $r_{f}$, the dynamics of the log-spot $X_{t}$ reads
\begin{eqnarray}
dX_{t} & = & -\frac{V_{t}^{2}}{2}dt+V_{t}dW_{t},\quad X_{0}=x_{0}\nonumber \\
dV_{t} & = & \kappa_{t}(\theta_{t}-V_{t})dt+\lambda_{t}V_{t}dB_{t},\quad V_{0}=v_{0}\label{eq:logX_V}\\
d\left\langle W,B\right\rangle _{t} & = & \rho_{t}dt\,,\nonumber 
\end{eqnarray}
where $x_{0}=\log\left(S_{0}\right)$. Define a perturbed process
$\left(X^{\varepsilon},V^{\varepsilon}\right)$ as follows
\begin{eqnarray}
dX_{t}^{\varepsilon} & = & -\frac{(V_{t}^{\varepsilon})^{2}}{2}dt+V_{t}^{\varepsilon}dW_{t},\quad X_{0}^{\varepsilon}=x_{0}\nonumber \\
dV_{t}^{\varepsilon} & = & \kappa_{t}(\theta_{t}-V_{t}^{\varepsilon})dt+\varepsilon\lambda_{t}V_{t}^{\varepsilon}dB_{t},\quad V_{0}^{\varepsilon}=v_{0}\label{eq:logX_Veps}\\
d\left\langle W,B\right\rangle _{t} & = & \rho_{t}dt\,,\nonumber 
\end{eqnarray}
and define
\begin{equation}
g(\varepsilon)=\exp\left(-\int_{0}^{T}r_{d}(t)dt\right)\mathbb{E}\left[\left(K-\exp\left(-\int_{0}^{T}\left(r_{d}(t)-r_{f}(t)\right)dt\right)+X_{T}^{\varepsilon}\right)_{+}\right]\,,\label{eq:g_Exp}
\end{equation}
so that $g(1)=P_{IGa}$, the price of European put with IGa stochastic
volatility, which is the quantity we want to compute. Remark that
$g(0)$ reduces to a Black-Scholes price.

The expression \eqref{eq:g_Exp} for $g(\varepsilon)$ can be simplified.
Remark that $W_{t}$ can be decomposed into
\[
W_{t}=\rho_{t}B_{t}+\sqrt{1-\rho_{t}^{2}}dB_{t}^{\perp}
\]
where $B^{\bot}$ is a Brownian motion independent from $B$. Therefore
\begin{eqnarray*}
X_{t}^{\varepsilon} & = & x_{0}+\int_{0}^{T}\rho_{t}V_{t}^{\varepsilon}dB_{t}-\frac{1}{2}\int_{0}^{T}(V_{t}^{\varepsilon})^{2}dt+\int_{0}^{T}\sqrt{1-\rho_{t}^{2}}V_{t}^{\varepsilon}dB_{t}^{\perp}\,.
\end{eqnarray*}
Let $\mathcal{F}^{B}=\left(\mathcal{F}_{t}^{B}\right)_{0\leq t\leq T}$
be the filtration generated by $B$. One can see that $X_{T}^{\varepsilon}\left|\mathcal{F}_{T}^{B}\right.$
has a Gaussian distribution with mean $x_{0}+\int_{0}^{T}\rho_{t}V_{t}^{\varepsilon}dB_{t}-\frac{1}{2}\int_{0}^{T}(V_{t}^{\varepsilon})^{2}dt=\left\{ x_{0}+\int_{0}^{T}\rho_{t}V_{t}^{\varepsilon}dB_{t}-\frac{1}{2}\int_{0}^{T}(\rho_{t}V_{t}^{\varepsilon})^{2}dt\right\} -\frac{1}{2}\int_{0}^{T}(1-\rho_{t}^{2})(V_{t}^{\varepsilon})^{2}dt$
and variance $\int_{0}^{T}(1-\rho_{t}^{2})(V_{t}^{\varepsilon})^{2}dt$.
Consequently,
\begin{eqnarray}
g(\varepsilon) & = & \exp\left(-\int_{0}^{T}r_{d}(t)dt\right)\mathbb{E}\left[\mathbb{E}\left[\left(K-\exp\left(-\int_{0}^{T}\left(r_{d}(t)-r_{f}(t)\right)dt\right)+X_{T}^{\varepsilon}\right)_{+}\left|\mathcal{F}_{T}^{B}\right.\right]\right]\nonumber \\
 & = & \mathbb{E}\left[P_{BS}\left(x_{0}+\int_{0}^{T}\rho_{t}V_{t}^{\varepsilon}dB_{t}-\frac{1}{2}\int_{0}^{T}(\rho_{t}V_{t}^{\varepsilon})^{2}dt\,,\,\int_{0}^{T}(1-\rho_{t}^{2})(V_{t}^{\varepsilon})^{2}dt\right)\right]\label{eq:g_PBS}
\end{eqnarray}

where $P_{BS}(x,y)$ is the Black-Scholes price of a European put
option with spot $e^{x}$ and integrated variance $y$ (equation \eqref{eq:PBS}).

For any non-negative integer $n$, define $V_{n,t}^{\varepsilon}:=\frac{\partial^{n}V_{t}^{\varepsilon}}{\partial\varepsilon^{n}}$
and $V_{n,t}:=\left.\frac{\partial^{n}V_{t}^{\varepsilon}}{\partial\varepsilon^{n}}\right|_{\varepsilon=0}$.
One can check that
\begin{eqnarray*}
dV_{0,t}^{\varepsilon} & = & \kappa_{t}(\theta_{t}-V_{0,t}^{\varepsilon})dt+\varepsilon\lambda_{t}V_{0,t}^{\varepsilon}dB_{t},\quad V_{0,0}^{\varepsilon}=v_{0}\\
dV_{n,t}^{\varepsilon} & = & -\kappa_{t}V_{n,t}^{\varepsilon}dt+n\lambda_{t}V_{n-1,t}^{\varepsilon}dB_{t}+\varepsilon\lambda_{t}V_{n,t}^{\varepsilon}dB_{t},\quad V_{n,0}^{\varepsilon}=0,\quad n\geq1
\end{eqnarray*}
and
\begin{eqnarray}
v_{0,t} & = & e^{-\int_{0}^{t}\kappa_{z}dz}\left(v_{0}+\int_{0}^{t}\kappa_{s}\theta_{s}e^{\int_{0}^{s}\kappa_{z}dz}ds\right)\label{eq:v0}\\
V_{n,t} & = & e^{-\int_{0}^{t}\kappa_{z}dz}\int_{0}^{t}e^{\int_{0}^{s}\kappa_{z}dz}n\lambda_{s}V_{n-1,s}dB_{s},\quad n\geq1\label{eq:Vn}
\end{eqnarray}
where $v_{0,t}:=V_{0,t}$ is deterministic.

At this point, the main idea of the proof is to approximate $g(1)$
using a Taylor expansion of $g(\varepsilon)$ around $\varepsilon=0$,
as $g$ and its derivatives reduce to a Black-Scholes formula when
$\varepsilon=0$. 

Thus, we approximate $V_{t}=V_{t}^{1}$ ($\varepsilon=1$) using $V_{t}^{0}=v_{0,t}$
($\varepsilon=0$) by applying the Taylor formula to the function
$\varepsilon\mapsto V_{t}^{\varepsilon}$
\[
V_{t}=v_{0,t}+V_{1,t}+\frac{1}{2}V_{2,t}+\ldots
\]
Doing the same to the function $\varepsilon\mapsto(V_{t}^{\varepsilon})^{2}$
yields
\[
(V_{t})^{2}=(v_{0,t})^{2}+2v_{0,t}V_{1,t}+(v_{0,t}V_{2,t}+(V_{1,t})^{2})+\ldots
\]
To simplify notations, define, for $i\geq0$ and $j\geq0$,
\begin{equation}
\frac{\partial^{i+j}\tilde{P}_{BS}}{\partial x^{i}y^{j}}:=\frac{\partial^{i+j}P_{BS}}{\partial x^{i}y^{j}}\left(x_{0}+\int_{0}^{T}\rho_{t}v_{0,t}dB_{t}-\frac{1}{2}\int_{0}^{T}\rho_{t}^{2}v_{0,t}^{2}dt\,,\,\int_{0}^{T}(1-\rho_{t}^{2})v_{0,t}^{2}dt\right)\,.\label{eq:PtildeBS}
\end{equation}
Then, the second-order Taylor expansion of $g(\varepsilon)$ around
$\varepsilon=0$, valued at $\varepsilon=1$, reads (keeping only second-order terms)
\begin{eqnarray}
g(1) & = & \mathbb{E}\left[\tilde{P}_{BS}\right]\label{eq:g(1)}\\
(C_{x}:=) & + & \mathbb{E}\left[\frac{\partial\tilde{P}_{BS}}{\partial x}\left\{ \int_{0}^{T}\rho_{t}\left(V_{1,t}+\frac{1}{2}V_{2,t}\right)dB_{t}-\frac{1}{2}\int_{0}^{T}\rho_{t}^{2}\left(2v_{0,t}V_{1,t}+V_{1,t}^{2}+v_{0,t}V_{2,t}\right)dt\right\} \right]\nonumber \\
(C_{y}:=) & + & \mathbb{E}\left[\frac{\partial\tilde{P}_{BS}}{\partial y}\left\{ \int_{0}^{T}\left(1-\rho_{t}^{2}\right)\left(2v_{0,t}V_{1,t}+V_{1,t}^{2}+v_{0,t}V_{2,t}\right)dt\right\} \right]\nonumber \\
(C_{xx}:=) & + & \frac{1}{2}\mathbb{E}\left[\frac{\partial^{2}\tilde{P}_{BS}}{\partial x^{2}}\left\{ \int_{0}^{T}\rho_{t}\left(V_{1,t}\right)dB_{t}-\frac{1}{2}\int_{0}^{T}\rho_{t}^{2}\left(2v_{0,t}V_{1,t}\right)dt\right\} ^{2}\right]\nonumber \\
(C_{yy}:=) & + & \frac{1}{2}\mathbb{E}\left[\frac{\partial^{2}\tilde{P}_{BS}}{\partial y^{2}}\left\{ \int_{0}^{T}\left(1-\rho_{t}^{2}\right)\left(2v_{0,t}V_{1,t}\right)dt\right\} ^{2}\right]\nonumber \\
(C_{xy}:=) & + & \mathbb{E}\left[\frac{\partial^{2}\tilde{P}_{BS}}{\partial x\partial y}\left\{ \int_{0}^{T}\rho_{t}\left(V_{1,t}\right)dB_{t}-\frac{1}{2}\int_{0}^{T}\rho_{t}^{2}\left(2v_{0,t}V_{1,t}\right)dt\right\} \left\{ \int_{0}^{T}\left(1-\rho_{t}^{2}\right)\left(2v_{0,t}V_{1,t}\right)dt\right\} \right]\nonumber \\
 & + & \mathcal{E}\,,\nonumber 
\end{eqnarray}
where $\mathcal{E}$ is the expansion error. From definition \eqref{eq:PtildeBS},
taking $\varepsilon=0$ in \eqref{eq:g_PBS} shows that
\begin{equation}
\mathbb{E}\left[\tilde{P}_{BS}\right]=P_{BS}\left(x_{0},\int_{0}^{T}v_{0,t}^{2}dt\right)\,.\label{eq:EPBS}
\end{equation}
Similarly,
\begin{equation}
\mathbb{E}\left[\frac{\partial^{i+j}\tilde{P}_{BS}}{\partial x^{i}y^{j}}\right]=\frac{\partial^{i+j}P_{BS}}{\partial x^{i}y^{j}}\left(x_{0},\int_{0}^{T}v_{0,t}^{2}dt\right)\,.\label{eq:EGPBS}
\end{equation}
The next part of this appendix is devoted to the computation of all
the constants $C_{x}$, $C_{y}$, $C_{xx}$, $C_{yy}$ and $C_{xy}$.
Our main tools will be the following relation between Black-Sholes
greeks:
\begin{equation}
\frac{\partial P_{BS}}{\partial y}\left(x,y\right)=\frac{1}{2}\left(\frac{\partial^{2}P_{BS}}{\partial x^{2}}\left(x,y\right)-\frac{\partial P_{BS}}{\partial x}\left(x,y\right)\right)\,,\label{eq:BS_identity}
\end{equation}
the following product identity between two stochastic processes $X$
and $Y$:
\begin{equation}
X_{T}Y_{T}=X_{0}Y_{0}+\int_{0}^{T}\left(X_{t}dY_{t}+Y_{t}dX_{t}+d\left\langle X,Y\right\rangle _{t}\right)\,,\label{eq:product_identity}
\end{equation}
and the following Lemma:
\begin{lem}
\label{lem:Malliavin_identity}If $G=l\left(\int_{0}^{T}\rho_{u}v_{0,u}dB_{u}\right)$
for a differentiable function $l$ with derivative $l^{(1)}$, then
its first Malliavin derivative $D^{B}\left(G\right)=\left(D_{s}^{B}\left(G\right)\right)_{s\geq0}$
is equal to $D_{s}^{B}\left(G\right)=l^{(1)}\left(\int_{0}^{T}\rho_{u}v_{0,u}dB_{u}\right)\rho_{s}v_{0,s}\mathds{1}\{s\leq T\}$.
Therefore, using Lemma 5.2 from \citet{Benhamou10},
\begin{equation}
\mathbb{E}\left[l\left(\int_{0}^{T}\rho_{u}v_{0,u}dB_{u}\right)\left(\int_{0}^{t}\alpha_{s}dB_{s}\right)\right]=\mathbb{E}\left[l^{(1)}\left(\int_{0}^{T}\rho_{u}v_{0,u}dB_{u}\right)\left(\int_{0}^{t}\rho_{s}v_{0,s}\alpha_{s}ds\right)\right]\label{eq:Malliavin_identity}
\end{equation}
for $0\leq t\leq T$. Most of the time, we will use \eqref{eq:Malliavin_identity}
with $t=T$.\end{lem}
\begin{itemize}
\item Computation of $C_{x}$:
\end{itemize}
Using \eqref{eq:Malliavin_identity},
\[
\mathbb{E}\left[\frac{\partial\tilde{P}_{BS}}{\partial x}\left\{ \int_{0}^{T}\rho_{t}\left(V_{1,t}+\frac{1}{2}V_{2,t}\right)dB_{t}\right\} \right]=\mathbb{E}\left[\frac{\partial^{2}\tilde{P}_{BS}}{\partial x^{2}}\left\{ \int_{0}^{T}\rho_{t}^{2}\left(v_{0,t}V_{1,t}+\frac{1}{2}v_{0,t}V_{2,t}\right)dt\right\} \right]\,.
\]
Therefore, using \eqref{eq:BS_identity},
\[
C_{x}=\mathbb{E}\left[\frac{\partial\tilde{P}_{BS}}{\partial y}\left\{ \int_{0}^{T}\rho_{t}^{2}\left(2v_{0,t}V_{1,t}+v_{0,t}V_{2,t}\right)dt\right\} \right]-\mathbb{E}\left[\frac{\partial\tilde{P}_{BS}}{\partial x}\left\{ \frac{1}{2}\int_{0}^{T}\rho_{t}^{2}V_{1,t}^{2}dt\right\} \right]
\]

\begin{itemize}
\item Computation of $C_{xx}$:
\end{itemize}
Using \eqref{eq:product_identity}, then \eqref{eq:Malliavin_identity},
and finally \eqref{eq:BS_identity}, 
\begin{eqnarray*}
C_{xx} & = & \mathbb{E}\left[\frac{\partial^{2}\tilde{P}_{BS}}{\partial x^{2}}\int_{0}^{T}\left\{ \int_{0}^{t}\rho_{s}V_{1,s}dB_{s}-\int_{0}^{t}\rho_{s}^{2}v_{0,s}V_{1,s}ds\right\} \left\{ \rho_{t}V_{1,t}dB_{t}-\rho_{t}^{2}v_{0,t}V_{1,t}dt\right\} \right]+\frac{1}{2}\mathbb{E}\left[\frac{\partial^{2}\tilde{P}_{BS}}{\partial x^{2}}\int_{0}^{T}\rho_{t}^{2}V_{1,t}^{2}dt\right]\\
 & = & \mathbb{E}\left[\left(\frac{\partial^{3}\tilde{P}_{BS}}{\partial x^{3}}-\frac{\partial^{2}\tilde{P}_{BS}}{\partial x^{2}}\right)\int_{0}^{T}\left\{ \int_{0}^{t}\rho_{s}V_{1,s}dB_{s}-\int_{0}^{t}\rho_{s}^{2}v_{0,s}V_{1,s}ds\right\} \rho_{t}^{2}v_{0,t}V_{1,t}dt\right]+\frac{1}{2}\mathbb{E}\left[\frac{\partial^{2}\tilde{P}_{BS}}{\partial x^{2}}\int_{0}^{T}\rho_{t}^{2}V_{1,t}^{2}dt\right]\\
 & = & 2\mathbb{E}\left[\frac{\partial^{2}\tilde{P}_{BS}}{\partial x\partial y}\int_{0}^{T}\left\{ \int_{0}^{t}\rho_{s}V_{1,s}dB_{s}-\int_{0}^{t}\rho_{s}^{2}v_{0,s}V_{1,s}ds\right\} \rho_{t}^{2}v_{0,t}V_{1,t}dt\right]+\frac{1}{2}\mathbb{E}\left[\frac{\partial^{2}\tilde{P}_{BS}}{\partial x^{2}}\int_{0}^{T}\rho_{t}^{2}V_{1,t}^{2}dt\right]
\end{eqnarray*}
Therefore, so far, using \eqref{eq:BS_identity},
\begin{eqnarray*}
C_{x}+C_{y}+C_{xx} & = & \mathbb{E}\left[\frac{\partial\tilde{P}_{BS}}{\partial y}\left\{ \int_{0}^{T}\left(2v_{0,t}V_{1,t}+V_{1,t}^{2}+v_{0,t}V_{2,t}\right)dt\right\} \right]\\
(\tilde{C}_{xy}:=) & + & 2\mathbb{E}\left[\frac{\partial^{2}\tilde{P}_{BS}}{\partial x\partial y}\int_{0}^{T}\left\{ \int_{0}^{t}\rho_{s}V_{1,s}dB_{s}-\int_{0}^{t}\rho_{s}^{2}v_{0,s}V_{1,s}ds\right\} \rho_{t}^{2}v_{0,t}V_{1,t}dt\right]
\end{eqnarray*}

\begin{itemize}
\item Computation of $C_{xy}$:
\end{itemize}
Using \eqref{eq:product_identity},
\begin{eqnarray*}
C_{xy} & = & \mathbb{E}\left[\frac{\partial^{2}\tilde{P}_{BS}}{\partial x\partial y}\left\{ \int_{0}^{T}\rho_{t}V_{1,t}dB_{t}-\int_{0}^{T}\rho_{t}^{2}v_{0,t}V_{1,t}dt\right\} \left\{ \int_{0}^{T}\left(1-\rho_{t}^{2}\right)\left(2v_{0,t}V_{1,t}\right)dt\right\} \right]\\
 & = & \mathbb{E}\left[\frac{\partial^{2}\tilde{P}_{BS}}{\partial x\partial y}\int_{0}^{T}\left(\int_{0}^{t}\rho_{s}V_{1,s}dB_{s}-\int_{0}^{t}\rho_{s}^{2}v_{0,s}V_{1,s}ds\right)\left(1-\rho_{t}^{2}\right)\left(2v_{0,t}V_{1,t}\right)dt\right]\\
 & + & \mathbb{E}\left[\frac{\partial^{2}\tilde{P}_{BS}}{\partial x\partial y}\int_{0}^{T}\left\{ \int_{0}^{t}\left(1-\rho_{s}^{2}\right)\left(2v_{0,s}V_{1,s}\right)ds\right\} \left\{ \rho_{t}V_{1,t}dB_{t}-\rho_{t}^{2}v_{0,t}V_{1,t}dt\right\} \right]
\end{eqnarray*}

Using \eqref{eq:Malliavin_identity},
\begin{eqnarray*}
C_{xy}+\tilde{C}_{xy} & = & \mathbb{E}\left[\frac{\partial^{2}\tilde{P}_{BS}}{\partial x\partial y}\int_{0}^{T}\left(\int_{0}^{t}\rho_{s}V_{1,s}dB_{s}-\int_{0}^{t}\rho_{s}^{2}v_{0,s}V_{1,s}ds\right)\left(2v_{0,t}V_{1,t}\right)dt\right]\\
 & + & \mathbb{E}\left[\frac{\partial^{2}\tilde{P}_{BS}}{\partial x\partial y}\int_{0}^{T}\left\{ \int_{0}^{t}\left(1-\rho_{s}^{2}\right)\left(2v_{0,s}V_{1,s}\right)ds\right\} \left\{ \rho_{t}V_{1,t}dB_{t}-\rho_{t}^{2}v_{0,t}V_{1,t}dt\right\} \right]
\end{eqnarray*}
Define $G=\frac{\partial^{2}\tilde{P}_{BS}}{\partial x\partial y}v_{0,t}V_{1,t}$.
Then $D_{s}^{B}\left(G\right)$, its first Malliavin derivative w.r.t.
$B$, is given by
\begin{eqnarray*}
D_{s}^{B}\left(G\right) & = & v_{0,t}V_{1,t}D_{s}^{B}\left(\frac{\partial^{2}\tilde{P}_{BS}}{\partial x\partial y}\right)+\frac{\partial^{2}\tilde{P}_{BS}}{\partial x\partial y}v_{0,t}D_{s}^{B}\left(V_{1,t}\right)\\
 & = & v_{0,t}V_{1,t}\frac{\partial^{3}\tilde{P}_{BS}}{\partial x^{2}\partial y}\rho_{s}v_{0,s}\mathds{1}\{s\leq T\}+\frac{\partial^{2}\tilde{P}_{BS}}{\partial x\partial y}v_{0,t}e^{-\int_{0}^{t}\kappa_{z}dz}e^{\int_{0}^{s}\kappa_{z}dz}\lambda_{s}v_{0,s}\mathds{1}\{s\leq t\}
\end{eqnarray*}
using the definition of $V_{1,t}$ \eqref{eq:Vn}. Therefore, using
Lemma 5.2 in \citet{Benhamou10},
\begin{eqnarray*}
 &  & 2\int_{0}^{T}\mathbb{E}\left[\frac{\partial^{2}\tilde{P}_{BS}}{\partial x\partial y}v_{0,t}V_{1,t}\left(\int_{0}^{t}\rho_{s}V_{1,s}dB_{s}\right)\right]dt\\
 & = & 2\int_{0}^{T}\mathbb{E}\left[G\left\{ \int_{0}^{t}\rho_{s}V_{1,s}dB_{s}\right\} \right]dt\\
 & = & 2\int_{0}^{T}\mathbb{E}\left[\int_{0}^{t}\rho_{s}V_{1,s}D_{s}^{B}\left(G\right)ds\right]dt\\
 & = & 2\mathbb{E}\left[\frac{\partial^{3}\tilde{P}_{BS}}{\partial x^{2}\partial y}\int_{0}^{T}\left(\int_{0}^{t}\rho_{s}^{2}v_{0,s}V_{1,s}ds\right)v_{0,t}V_{1,t}dt\right]+2\int_{0}^{T}v_{0,t}e^{-\int_{0}^{t}\kappa_{z}dz}\mathbb{E}\left[\frac{\partial^{2}\tilde{P}_{BS}}{\partial x\partial y}\int_{0}^{t}e^{\int_{0}^{s}\kappa_{z}dz}\lambda_{s}\rho_{s}v_{0,s}V_{1,s}ds\right]dt\\
 & = & 2\mathbb{E}\left[\frac{\partial^{3}\tilde{P}_{BS}}{\partial x^{2}\partial y}\int_{0}^{T}\left(\int_{0}^{t}\rho_{s}^{2}v_{0,s}V_{1,s}ds\right)v_{0,t}V_{1,t}dt\right]+\mathbb{E}\left[\frac{\partial\tilde{P}_{BS}}{\partial y}\int_{0}^{T}v_{0,t}V_{2,t}dt\right]
\end{eqnarray*}
where we used, for the last equality, equation \eqref{eq:Malliavin_identity}
and the definition of $V_{2,t}$ (equation \eqref{eq:Vn}). Therefore,
using \eqref{eq:Malliavin_identity} and then \eqref{eq:BS_identity},
\begin{eqnarray*}
C_{xy}+\tilde{C}_{xy} & = & 4\mathbb{E}\left[\frac{\partial^{2}\tilde{P}_{BS}}{\partial y^{2}}\int_{0}^{T}\left(\int_{0}^{t}\rho_{s}^{2}v_{0,s}V_{1,s}ds\right)v_{0,t}V_{1,t}dt\right]\\
 & + & 4\mathbb{E}\left[\frac{\partial^{2}\tilde{P}_{BS}}{\partial y^{2}}\int_{0}^{T}\left\{ \int_{0}^{t}\left(1-\rho_{s}^{2}\right)v_{0,s}V_{1,s}ds\right\} \rho_{t}^{2}v_{0,t}V_{1,t}dt\right]\\
 & + & \mathbb{E}\left[\frac{\partial\tilde{P}_{BS}}{\partial y}\int_{0}^{T}v_{0,t}V_{2,t}dt\right]
\end{eqnarray*}
Using \eqref{eq:product_identity},
\begin{eqnarray*}
C_{yy} & = & 4\mathbb{E}\left[\frac{\partial^{2}\tilde{P}_{BS}}{\partial y^{2}}\int_{0}^{T}\left\{ \int_{0}^{t}\left(1-\rho_{s}^{2}\right)v_{0,s}V_{1,s}ds\right\} \left(1-\rho_{t}^{2}\right)v_{0,t}V_{1,t}dt\right]\,.
\end{eqnarray*}
Thus

\begin{eqnarray*}
C_{xy}+\tilde{C}_{xy}+C_{yy} & = & 4\mathbb{E}\left[\frac{\partial^{2}\tilde{P}_{BS}}{\partial y^{2}}\int_{0}^{T}\left(\int_{0}^{t}v_{0,s}V_{1,s}ds\right)v_{0,t}V_{1,t}dt\right]+\mathbb{E}\left[\frac{\partial\tilde{P}_{BS}}{\partial y}\int_{0}^{T}v_{0,t}V_{2,t}dt\right]\\
 & = & 2\mathbb{E}\left[\frac{\partial^{2}\tilde{P}_{BS}}{\partial y^{2}}\left\{ \int_{0}^{T}v_{0,t}V_{1,t}dt\right\} ^{2}\right]+\mathbb{E}\left[\frac{\partial\tilde{P}_{BS}}{\partial y}\int_{0}^{T}v_{0,t}V_{2,t}dt\right]\,.
\end{eqnarray*}
At this stage, combining all the terms together, the second-order
expansion for $g(1)$ (equation \eqref{eq:g(1)}) simply becomes
\begin{align}
g(1) & =P_{BS}\left(x_{0},\int_{0}^{T}v_{0,t}^{2}dt\right)+\mathbb{E}\left[\frac{\partial\tilde{P}_{BS}}{\partial y}\int_{0}^{T}\left(2v_{0,t}V_{1,t}+V_{1,t}^{2}+2v_{0,t}V_{2,t}\right)dt\right]+2\mathbb{E}\left[\frac{\partial^{2}\tilde{P}_{BS}}{\partial y^{2}}\left\{ \int_{0}^{T}v_{0,t}V_{1,t}dt\right\} ^{2}\right]+\mathcal{E}\label{eq:g(1)simple}
\end{align}
Equation \eqref{eq:g(1)simple} is much simpler than equation \eqref{eq:g(1)},
but still contains the stochastic processes $V_{1,t}$ and $V_{2,t}$.
The last part of the proof is to simplify \eqref{eq:g(1)simple} further
to get an expansion with explicit, deterministic coefficients. 

Recall the short-hand notation for deterministic integrals from equations
\eqref{eq:omega_0} and \eqref{eq:omega_n}:

\begin{eqnarray*}
\omega_{t,T}^{\left(\kappa,l\right)}=\int_{t}^{T}e^{\int_{0}^{u}\kappa_{z}dz}l_{u}du &  & \forall t\in\left[0,T\right]
\end{eqnarray*}
\begin{eqnarray*}
\omega_{t,T}^{\left(\kappa_{n},l_{n}\right),\ldots,\left(\kappa_{1},l_{1}\right)}=\omega_{t,T}^{\left(\kappa_{n},l_{n}\omega_{.,T}^{\left(\kappa_{n-1},l_{n-1}\right),\ldots,\left(\kappa_{1},l_{1}\right)}\right)} &  & \forall t\in\left[0,T\right]\,.
\end{eqnarray*}

We will use extensively the following Lemma:
\begin{lem}
\label{lem:IPP}(Lemma 5.4 in \citet{Benhamou10}) For any deterministic
integrable function $f$ and any continuous semimartingale $Z$ such
that $Z_{0}=0$,
\[
\int_{0}^{T}f(t)Z_{t}dt=\int_{0}^{T}\omega_{t,T}^{(0,f)}dZ_{t}
\]
\end{lem}
\begin{proof}
Apply It\={o}'s lemma to the product $\omega_{t,T}^{(0,f)}Z_{t}$.
\end{proof}
The computation of each type of expectation in equation \eqref{eq:g(1)simple}
is summarized in the following Lemma.
\begin{lem}
\label{lem:determinify}The following equalities hold
\begin{align}
\mathbb{E}\left[l\left(\int_{0}^{T}\rho_{t}v_{0,t}dB_{t}\right)\int_{0}^{T}\beta_{t}V_{1,t}dt\right] & =\omega_{0,T}^{(\kappa,\rho\lambda v_{0,.}^{2}),(-\kappa,\beta)}\mathbb{E}\left[l^{(1)}\left(\int_{0}^{T}\rho_{t}v_{0,t}dB_{t}\right)\right]\label{eq:lV1}\\
\mathbb{E}\left[l\left(\int_{0}^{T}\rho_{t}v_{0,t}dB_{t}\right)\int_{0}^{T}\beta_{t}V_{2,t}dt\right] & =\omega_{0,T}^{(\kappa,\rho\lambda v_{0,.}^{2}),(0,2\rho\lambda v_{0,.}),(-\kappa,\beta)}\mathbb{E}\left[l^{(2)}\left(\int_{0}^{T}\rho_{t}v_{0,t}dB_{t}\right)\right]\label{eq:lV2}\\
\mathbb{E}\left[l\left(\int_{0}^{T}\rho_{t}v_{0,t}dB_{t}\right)\int_{0}^{T}\beta_{t}V_{1,t}^{2}dt\right] & =\omega_{0,T}^{(2\kappa,\lambda^{2}v_{0,.}^{2}),(-2\kappa,\beta)}\mathbb{E}\left[l\left(\int_{0}^{T}\rho_{t}v_{0,t}dB_{t}\right)\right]\nonumber \\
 & +2\omega_{0,T}^{(\kappa,\rho\lambda v_{0,.}^{2}),(\kappa,\rho\lambda v_{0,.}^{2}),(-2\kappa,\beta)}\mathbb{E}\left[l^{(2)}\left(\int_{0}^{T}\rho_{t}v_{0,t}dB_{t}\right)\right]\label{eq:lV12}\\
\mathbb{E}\left[l\left(\int_{0}^{T}\rho_{t}v_{0,t}dB_{t}\right)\left\{ \int_{0}^{T}\beta_{t}V_{1,t}dt\right\} ^{2}\right] & =2\omega_{0,T}^{(2\kappa,\lambda^{2}v_{0,.}^{2}),(-\kappa,\beta),(-\kappa,\beta)}\mathbb{E}\left[l\left(\int_{0}^{T}\rho_{t}v_{0,t}dB_{t}\right)\right]\nonumber \\
 & +\left\{ \omega_{0,T}^{(\kappa,\rho\lambda v_{0,.}^{2}),(-\kappa,\beta)}\right\} ^{2}\mathbb{E}\left[l^{(2)}\left(\int_{0}^{T}\rho_{t}v_{0,t}dB_{t}\right)\right]\label{eq:lV1_2}
\end{align}
where $\beta$, a deterministic function, and $l$, a twice-differentiable
function, are such that these expectations exist.\end{lem}
\begin{proof}
Using Lemma \ref{lem:IPP} and equation \eqref{eq:Malliavin_identity},
\begin{align*}
\mathbb{E}\left[l\left(\int_{0}^{T}\rho_{t}v_{0,t}dB_{t}\right)\int_{0}^{T}\beta_{t}V_{1,t}dt\right] & =\mathbb{E}\left[l\left(\int_{0}^{T}\rho_{t}v_{0,t}dB_{t}\right)\int_{0}^{T}\beta_{t}e^{-\int_{0}^{t}\kappa_{z}dz}\int_{0}^{t}e^{\int_{0}^{s}\kappa_{z}dz}\lambda_{s}v_{0,s}dB_{s}dt\right]\\
 & =\mathbb{E}\left[l\left(\int_{0}^{T}\rho_{t}v_{0,t}dB_{t}\right)\int_{0}^{T}\omega_{t,T}^{(-\kappa,\beta)}e^{\int_{0}^{t}\kappa_{z}dz}\lambda_{t}v_{0,t}dB_{t}\right]\\
 & =\omega_{0,T}^{(\kappa,\rho\lambda v_{0,.}^{2}),(-\kappa,\beta)}\mathbb{E}\left[l^{(1)}\left(\int_{0}^{T}\rho_{t}v_{0,t}dB_{t}\right)\right]\,,
\end{align*}
\begin{align*}
 & \mathbb{E}\left[l\left(\int_{0}^{T}\rho_{t}v_{0,t}dB_{t}\right)\int_{0}^{T}\beta_{t}V_{2,t}dt\right]=\mathbb{E}\left[l\left(\int_{0}^{T}\rho_{t}v_{0,t}dB_{t}\right)\int_{0}^{T}\beta_{t}e^{-\int_{0}^{t}\kappa_{z}dz}\int_{0}^{t}e^{\int_{0}^{s}\kappa_{z}dz}2\lambda_{s}V_{1,s}dB_{s}dt\right]\\
 & =\mathbb{E}\left[l\left(\int_{0}^{T}\rho_{t}v_{0,t}dB_{t}\right)\int_{0}^{T}\omega_{t,T}^{(-\kappa,\beta)}e^{\int_{0}^{t}\kappa_{z}dz}2\lambda_{t}V_{1,t}dB_{t}\right]=\mathbb{E}\left[l^{(1)}\left(\int_{0}^{T}\rho_{t}v_{0,t}dB_{t}\right)\int_{0}^{T}\omega_{t,T}^{(-\kappa,\beta)}e^{\int_{0}^{t}\kappa_{z}dz}2\rho_{t}\lambda_{t}v_{0,t}V_{1,t}dt\right]\\
 & =\mathbb{E}\left[l^{(1)}\left(\int_{0}^{T}\rho_{t}v_{0,t}dB_{t}\right)\int_{0}^{T}\omega_{t,T}^{(-\kappa,\beta)}2\rho_{t}\lambda_{t}v_{0,t}\int_{0}^{t}e^{\int_{0}^{s}\kappa_{z}dz}\lambda_{s}v_{0,s}dB_{s}dt\right]\\
 & =\mathbb{E}\left[l^{(1)}\left(\int_{0}^{T}\rho_{t}v_{0,t}dB_{t}\right)\int_{0}^{T}\omega_{t,T}^{(0,2\rho\lambda v_{0,.}),(-\kappa,\beta)}e^{\int_{0}^{t}\kappa_{z}dz}\lambda_{t}v_{0,t}dB_{t}\right]=\omega_{0,T}^{(\kappa,\rho\lambda v_{0,.}^{2}),(0,2\rho\lambda v_{0,.}),(-\kappa,\beta)}\mathbb{E}\left[l^{(2)}\left(\int_{0}^{T}\rho_{t}v_{0,t}dB_{t}\right)\right]\,.
\end{align*}
Using \eqref{eq:product_identity},
\begin{align*}
 & \mathbb{E}\left[l\left(\int_{0}^{T}\rho_{t}v_{0,t}dB_{t}\right)\int_{0}^{T}\beta_{t}V_{1,t}^{2}dt\right]=\mathbb{E}\left[l\left(\int_{0}^{T}\rho_{t}v_{0,t}dB_{t}\right)\int_{0}^{T}\beta_{t}e^{-2\int_{0}^{t}\kappa_{z}dz}\left\{ \int_{0}^{t}e^{\int_{0}^{s}\kappa_{z}dz}\lambda_{s}v_{0,s}dB_{s}\right\} ^{2}dt\right]\\
 & =\mathbb{E}\left[l\left(\int_{0}^{T}\rho_{t}v_{0,t}dB_{t}\right)\int_{0}^{T}\beta_{t}e^{-2\int_{0}^{t}\kappa_{z}dz}\int_{0}^{t}2\left\{ \int_{0}^{s}e^{\int_{0}^{u}\kappa_{z}dz}\lambda_{u}v_{0,u}dB_{u}\right\} e^{\int_{0}^{s}\kappa_{z}dz}\lambda_{s}v_{0,s}dB_{s}dt\right]\\
 & +\mathbb{E}\left[l\left(\int_{0}^{T}\rho_{t}v_{0,t}dB_{t}\right)\int_{0}^{T}\beta_{t}e^{-2\int_{0}^{t}\kappa_{z}dz}\left\{ \int_{0}^{t}e^{2\int_{0}^{s}\kappa_{z}dz}\lambda_{s}^{2}v_{0,s}^{2}ds\right\} dt\right]\,.
\end{align*}
First, using Lemma \ref{lem:IPP},
\begin{align*}
 & \mathbb{E}\left[l\left(\int_{0}^{T}\rho_{t}v_{0,t}dB_{t}\right)\int_{0}^{T}\beta_{t}e^{-2\int_{0}^{t}\kappa_{z}dz}\left\{ \int_{0}^{t}e^{2\int_{0}^{s}\kappa_{z}dz}\lambda_{s}^{2}v_{0,s}^{2}ds\right\} dt\right]=\omega_{0,T}^{(2\kappa,\lambda^{2}v_{0,.}^{2}),(-2\kappa,\beta)}\mathbb{E}\left[l\left(\int_{0}^{T}\rho_{t}v_{0,t}dB_{t}\right)\right]\,.
\end{align*}
Then, using Lemma \ref{lem:IPP}, Lemma \eqref{lem:Malliavin_identity}
and then equation \eqref{eq:lV1},
\begin{align*}
 & \mathbb{E}\left[l\left(\int_{0}^{T}\rho_{t}v_{0,t}dB_{t}\right)\int_{0}^{T}\beta_{t}e^{-2\int_{0}^{t}\kappa_{z}dz}\int_{0}^{t}2\left\{ \int_{0}^{s}e^{\int_{0}^{u}\kappa_{z}dz}\lambda_{u}v_{0,u}dB_{u}\right\} e^{\int_{0}^{s}\kappa_{z}dz}\lambda_{s}v_{0,s}dB_{s}dt\right]\\
 & =\mathbb{E}\left[l\left(\int_{0}^{T}\rho_{t}v_{0,t}dB_{t}\right)\int_{0}^{T}2\omega_{t,T}^{(-2\kappa,\beta)}e^{\int_{0}^{t}\kappa_{z}dz}\lambda_{t}v_{0,t}\left\{ \int_{0}^{t}e^{\int_{0}^{s}\kappa_{z}dz}\lambda_{s}v_{0,s}dB_{s}\right\} dB_{t}\right]\\
 & =\mathbb{E}\left[l^{(1)}\left(\int_{0}^{T}\rho_{t}v_{0,t}dB_{t}\right)\int_{0}^{T}2\omega_{t,T}^{(-2\kappa,\beta)}e^{2\int_{0}^{t}\kappa_{z}dz}\rho_{t}\lambda_{t}v_{0,t}^{2}V_{1,t}dt\right]\\
 & =2\omega_{0,T}^{(\kappa,\rho\lambda v_{0,.}^{2}),(\kappa,\rho\lambda v_{0,.}^{2}),(-2\kappa,\beta)}\mathbb{E}\left[l^{(2)}\left(\int_{0}^{T}\rho_{t}v_{0,t}dB_{t}\right)\right]\,,
\end{align*}
yielding \eqref{eq:lV12}. Finally, using \eqref{eq:product_identity}
and Lemma \ref{lem:IPP},
\begin{align*}
\mathbb{E}\left[l\left(\int_{0}^{T}\rho_{t}v_{0,t}dB_{t}\right)\left\{ \int_{0}^{T}\beta_{t}V_{1,t}dt\right\} ^{2}\right] & =\mathbb{E}\left[l\left(\int_{0}^{T}\rho_{t}v_{0,t}dB_{t}\right)2\int_{0}^{T}\left\{ \int_{0}^{t}\beta_{s}V_{1,s}ds\right\} \beta_{t}V_{1,t}dt\right]\\
 & =2\mathbb{E}\left[l\left(\int_{0}^{T}\rho_{t}v_{0,t}dB_{t}\right)\int_{0}^{T}\omega_{t,T}^{(-\kappa,\beta)}e^{\int_{0}^{t}\kappa_{z}dz}\beta_{t}V_{1,t}^{2}dt\right]\\
 & +2\mathbb{E}\left[l\left(\int_{0}^{T}\rho_{t}v_{0,t}dB_{t}\right)\int_{0}^{T}\omega_{t,T}^{(-\kappa,\beta)}\left\{ \int_{0}^{t}\beta_{s}V_{1,s}ds\right\} e^{\int_{0}^{t}\kappa_{z}dz}\lambda_{t}v_{0,t}dB_{t}\right]
\end{align*}
Using equation \eqref{eq:lV12},
\begin{align*}
2\mathbb{E}\left[l\left(\int_{0}^{T}\rho_{t}v_{0,t}dB_{t}\right)\int_{0}^{T}\omega_{t,T}^{(-\kappa,\beta)}e^{\int_{0}^{t}\kappa_{z}dz}\beta_{t}V_{1,t}^{2}dt\right] & =2\omega_{0,T}^{(2\kappa,\lambda^{2}v_{0,.}^{2}),(-\kappa,\beta),(-\kappa,\beta)}\mathbb{E}\left[l\left(\int_{0}^{T}\rho_{t}v_{0,t}dB_{t}\right)\right]\\
 & +4\omega_{0,T}^{(\kappa,\rho\lambda v_{0,.}^{2}),(\kappa,\rho\lambda v_{0,.}^{2}),(-\kappa,\beta),(-\kappa,\beta)}\mathbb{E}\left[l^{(2)}\left(\int_{0}^{T}\rho_{t}v_{0,t}dB_{t}\right)\right]\,.
\end{align*}
Then, using Lemma \ref{lem:Malliavin_identity}, Lemma \ref{lem:IPP},
and equation \eqref{eq:lV1}, 
\begin{align*}
 & 2\mathbb{E}\left[l\left(\int_{0}^{T}\rho_{t}v_{0,t}dB_{t}\right)\int_{0}^{T}\omega_{t,T}^{(-\kappa,\beta)}\left\{ \int_{0}^{t}\beta_{s}V_{1,s}ds\right\} e^{\int_{0}^{t}\kappa_{z}dz}\lambda_{t}v_{0,t}dB_{t}\right]\\
 & =2\mathbb{E}\left[l^{(1)}\left(\int_{0}^{T}\rho_{t}v_{0,t}dB_{t}\right)\int_{0}^{T}\omega_{t,T}^{(-\kappa,\beta)}e^{\int_{0}^{t}\kappa_{z}dz}\rho_{t}\lambda_{t}v_{0,t}^{2}\left\{ \int_{0}^{t}\beta_{s}V_{1,s}ds\right\} dt\right]\\
 & =2\mathbb{E}\left[l^{(1)}\left(\int_{0}^{T}\rho_{t}v_{0,t}dB_{t}\right)\int_{0}^{T}\omega_{t,T}^{(\kappa,\rho\lambda v_{0,.}^{2}),(-\kappa,\beta)}\beta_{t}V_{1,t}dt\right]\\
 & =2\omega_{0,T}^{(\kappa,\rho\lambda v_{0,.}^{2}),(-\kappa,\beta),(\kappa,\rho\lambda v_{0,.}^{2}),(-\kappa,\beta)}\mathbb{E}\left[l^{(2)}\left(\int_{0}^{T}\rho_{t}v_{0,t}dB_{t}\right)\right]\,,
\end{align*}
and using that $4\omega_{0,T}^{(\kappa,\rho\lambda v_{0,.}^{2}),(\kappa,\rho\lambda v_{0,.}^{2}),(-\kappa,\beta),(-\kappa,\beta)}+2\omega_{0,T}^{(\kappa,\rho\lambda v_{0,.}^{2}),(-\kappa,\beta),(\kappa,\rho\lambda v_{0,.}^{2}),(-\kappa,\beta)}=\left\{ \omega_{0,T}^{(\kappa,\rho\lambda v_{0,.}^{2}),(-\kappa,\beta)}\right\} ^{2}$
(\citet{Benhamou10} p. 34) yields equation \eqref{eq:lV1_2}. 

To conclude the proof, apply Lemma \ref{lem:determinify} and equation
\eqref{eq:EGPBS} to each term of equation \eqref{eq:g(1)simple},
yielding
\begin{eqnarray*}
g(1) & = & P_{BS}\left(x_{0},\int_{0}^{T}v_{0,t}^{2}dt\right)+\mathbb{E}\left[\frac{\partial\tilde{P}_{BS}}{\partial y}\int_{0}^{T}\left(2v_{0,t}V_{1,t}+V_{1,t}^{2}+2v_{0,t}V_{2,t}\right)dt\right]+2\mathbb{E}\left[\frac{\partial^{2}\tilde{P}_{BS}}{\partial y^{2}}\left\{ \int_{0}^{T}v_{0,t}V_{1,t}dt\right\} ^{2}\right]\\
 & = & P_{BS}\left(x_{0},\int_{0}^{T}v_{0,t}^{2}dt\right)\\
 & + & \omega_{0,T}^{(\kappa,\rho\lambda v_{0,.}^{2}),(-\kappa,2v_{0,.})}\frac{\partial^{2}P_{BS}}{\partial x\partial y}\left(x_{0},\int_{0}^{T}v_{0,t}^{2}dt\right)\\
 & + & \omega_{0,T}^{(2\kappa,\lambda^{2}v_{0,.}^{2}),(-2\kappa,1)}\frac{\partial P_{BS}}{\partial y}\left(x_{0},\int_{0}^{T}v_{0,t}^{2}dt\right)\\
 & + & 2\omega_{0,T}^{(\kappa,\rho\lambda v_{0,.}^{2}),(\kappa,\rho\lambda v_{0,.}^{2}),(-2\kappa,1)}\frac{\partial^{3}P_{BS}}{\partial x^{2}\partial y}\left(x_{0},\int_{0}^{T}v_{0,t}^{2}dt\right)\\
 & + & \omega_{0,T}^{(\kappa,\rho\lambda v_{0,.}^{2}),(0,2\rho\lambda v_{0,.}),(-\kappa,2v_{0,.})}\frac{\partial^{3}P_{BS}}{\partial x^{2}\partial y}\left(x_{0},\int_{0}^{T}v_{0,t}^{2}dt\right)\\
 & + & 4\omega_{0,T}^{(2\kappa,\lambda^{2}v_{0,.}^{2}),(-\kappa,v_{0,.}),(-\kappa,v_{0,.})}\frac{\partial^{2}P_{BS}}{\partial y^{2}}\left(x_{0},\int_{0}^{T}v_{0,t}^{2}dt\right)\\
 & + & 2\left\{ \omega_{0,T}^{(\kappa,\rho\lambda v_{0,.}^{2}),(-\kappa,v_{0,.})}\right\} ^{2}\frac{\partial^{4}P_{BS}}{\partial x^{2}\partial y^{2}}\left(x_{0},\int_{0}^{T}v_{0,t}^{2}dt\right)\\
 & + & \mathcal{E}\,,
\end{eqnarray*}
which corresponds to the announced Theorem \ref{thm:IGa_expansion}.\end{proof}

\end{document}